\documentclass[sigconf, nonacm, pdfa]{acmart}
\usepackage[a-2b]{pdfx}
\usepackage{graphics}
\usepackage{booktabs}
\usepackage{datetime}
\usepackage[linesnumbered,ruled,vlined]{algorithm2e}
\usepackage{thmtools, thm-restate}
\newcommand{\DAF}{\textsf{DAF}\xspace}
\newcommand{\VEQ}{\textsf{VEQ}\xspace}

\newcommand{\VC}{\textsf{VC}\xspace}
\newcommand{\GQL}{\textsf{GQL}\xspace}

\newcommand{\Alley}{\textsf{Alley}\xspace}
\newcommand{\AlleyPlus}{\textsf{Alley+}\xspace}
\newcommand{\AlleyTPI}{\textsf{Alley+TPI}\xspace}
\newcommand{\NeurSC}{\textsf{NeurSC}\xspace}
\newcommand{\LSS}{\textsf{LSS}\xspace}

\newcommand{\Fastest}{\texttt{FaST\textsubscript{est}}\xspace}

\newcommand{\Yeast}{\textsf{Yeast}\xspace}
\newcommand{\WordNet}{\textsf{WordNet}\xspace}

\newcommand{\Dblp}{\textsf{DBLP}\xspace}

\newcommand{\Youtube}{\textsf{Youtube}\xspace}
\newcommand{\Patents}{\textsf{Patents}\xspace}
\newcommand{\Yago}{\textsf{YAGO}\xspace}

\newcommand{\Prob}[1]{\mathbb{P}\left[#1\right]}
\newcommand{\Expectation}[1]{\mathbb{E}\left[#1\right]}

\usepackage{ifthen}

\usepackage{mathtools}
\DeclarePairedDelimiterX\set[1]\lbrace\rbrace{#1}
\DeclarePairedDelimiterX{\pair}[2]{\langle}{\rangle}{#1, #2}

\newtheoremstyle{MyStyle}%
  {2.5pt} %
  {2.5pt} %
  {} %
  {} %
  {\bfseries} %
  {.} %
  {.5em} %
  {} %

\usepackage{mathtools}
\usepackage{enumitem}
\usepackage[nameinlink]{cleveref}
\usepackage{caption}
\usepackage{subcaption}
\usepackage{multirow}
\usepackage{setspace}
\usepackage{lipsum}
\usepackage{multicol}
\setlist{leftmargin=4mm}
\DeclarePairedDelimiter\abs{\lvert}{\rvert}%
\renewcommand*{\paragraph}[1]{ \stepcounter{paragraph}\noindent \underline{\textbf{#1}}. }

\theoremstyle{MyStyle}
\newtheorem{theorem}{Theorem}[section]
\newtheorem{theorem*}{Theorem}
\newtheorem*{definition*}{Definition}
\newtheorem{definition}[theorem]{Definition}
\newtheorem{lemma}[theorem]{Lemma}
\newtheorem{example}[theorem]{Example}
\newtheorem*{prf*}{Proof}
\renewenvironment{proof}{\noindent \textbf{Proof.}\ }{\hfill$\square$\vspace{0.3em}}
\everypar{\looseness=-1}

\usepackage{balance}
\newcommand\vldbpagestyle{plain} 

\allowdisplaybreaks
\begin{document}
\raggedbottom

\author{Wonseok Shin}
\affiliation{%
  \institution{Seoul National University}
}
\email{wsshin@theory.snu.ac.kr}

\author{Siwoo Song}
\affiliation{%
  \institution{Seoul National University}
}
\email{swsong@theory.snu.ac.kr}

\author{Kunsoo Park}
\affiliation{%
  \institution{Seoul National University}
}
\authornote{Corresponding author}

\email{kpark@theory.snu.ac.kr}

\author{Wook-Shin Han}
\affiliation{%
  \institution{POSTECH}
}
\email{wshan@dblab.postech.ac.kr}
\title{Cardinality Estimation of Subgraph Matching:\\A Filtering-Sampling Approach}

\def\vldbtitle{Cardinality Estimation of Subgraph Matching: A Filtering-Sampling Approach}

\pagestyle{\vldbpagestyle}

\clearpage
\renewcommand{\thetable}{\arabic{table}}
\renewcommand{\thefigure}{\arabic{figure}}
\setcounter{section}{0}
\setcounter{figure}{0}
\setcounter{page}{1}
\setcounter{table}{0}
\begin{abstract}
Subgraph counting is a fundamental problem in understanding and analyzing graph structured data, yet computationally challenging. This calls for an accurate and efficient algorithm for Subgraph Cardinality Estimation, which is to estimate the number of all isomorphic embeddings of a query graph in a data graph.
We present \Fastest, a novel algorithm that combines 
(1) a powerful filtering technique to significantly reduce the sample space, 
(2) an adaptive tree sampling algorithm for accurate and efficient estimation, and 
(3) a worst-case optimal stratified graph sampling algorithm for hard instances.
Extensive experiments on real-world datasets show that \Fastest outperforms state-of-the-art sampling-based methods by up to two orders of magnitude and GNN-based methods by up to three orders of magnitude in terms of accuracy.

\end{abstract}

\maketitle

\section{Introduction}
Subgraph matching is the fundamental problem in understanding and analyzing graph structured data \cite{FundementalProblems}. Given a data graph and a query graph, subgraph matching is the problem of finding all isomorphic embeddings of the query graph in the data graph. Identifying the occurrences of specific subgraph patterns is crucial for various applications such as analyzing protein-protein interaction networks~\cite{SING, Bonnici2017}, revealing patterns and trends of user interactions in social networks~\cite{Fan2012}, and optimizing queries in relational databases~\cite{TurboHompp}.
Subgraph counting, the problem of counting the number of embeddings, is also of paramount importance in various applications, such as designing graph kernels~\cite{GraphletKernels, CycleKernel} and understanding biological networks~\cite{Bio1, Bio2}. 

\begin{figure}[t]
    \centering
    \includegraphics[width=\linewidth]{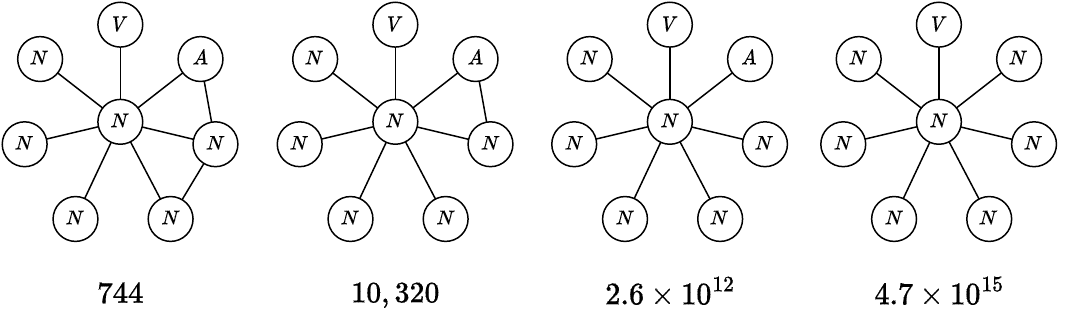}
    \caption{Query graphs and their number of embeddings in \WordNet dataset ($\approx$ 80K vertices). N, V, and A represent Noun, Verb, and Adjective, respectively.}
    \label{fig:Fig01}
\end{figure}
Both subgraph matching and subgraph counting are computationally challenging. Specifically, subgraph matching is NP-hard \cite{DAF}. Subgraph counting is \#P-hard, since subgraph counting in planar graphs (subgraph counting restricted to both data graph and query graph being planar) is \#P-hard \cite{PlanarCounting}.
This challenge is further amplified by the exponential number of possible embeddings.
For instance, \Cref{fig:Fig01} shows some query graphs with 8 vertices and the number of embeddings for each query graph in the \WordNet graph \cite{WordNet}. The number of embeddings varies from $744$ to $4.7\times 10^{15}$ with only small changes in the query graphs.
Consequently, numerous subgraph matching algorithms \cite{TurboIso, GraphQL, CFLMatch, DAF, VEQ}, which are primarily implemented for counting the number of embeddings rather than finding the actual embeddings, typically count at most $k$ (usually $10^3$ to $10^5$) embeddings, instead of counting all embeddings.
As a result, there is a pressing need for an accurate and efficient algorithm for \textit{Subgraph Cardinality Estimation}, which is to estimate the number of all isomorphic embeddings.

\paragraph{Existing Approaches and Limitations}
Subgraph cardinality estimation has been extensively researched due to its diverse real-world applications. In particular, the counting of small subgraphs (also known as graphlets or motifs) has been widely studied. IMPR \cite{IMPR} proposed a random walk-based method for counting graphlets, while MOTIVO \cite{MOTIVO} developed an adaptive graphlet sampling algorithm utilizing the color-coding technique. However, such works primarily focus on small subgraphs.

Recently, utilizing Graph Neural Networks (GNN) for subgraph cardinality estimation is gaining interest. LSS \cite{LSS} and NeurSC \cite{NeurSC} have explored GNNs for subgraph cardinality estimation. However, their accuracy for difficult instances with large queries remains unsatisfactory.

Approximate counting of homomorphic embeddings and estimating join cardinalities are also closely related problems to subgraph cardinality estimation.
Existing works in these areas mainly utilize summarization or sampling-based approaches. A comprehensive survey \cite{GCare} demonstrated that WanderJoin \cite{WanderJoin}, which employs random walks for sampling, outperforms other summarization and sampling methods for counting homomorphic embeddings. Alley \cite{Alley} extends WanderJoin by proposing random walk with intersection to improve accuracy. Despite the success of such methods, they often encounter sampling failure, particularly for instances with complex label distributions and large queries \cite{LSS}.

\paragraph{Contributions}
In this paper, we present a new algorithm \Fastest (\underline{\texttt{F}}iltering \underline{a}nd \underline{\texttt{S}}ampling \underline{\texttt{T}}echniques for subgraph cardinality \allowbreak \underline{\texttt{est}}-\allowbreak imation), addressing the limitations of existing methods.

We propose a novel \textit{Filtering-Sampling} approach for cardinality estimation as follows. To approximate the cardinality of a set whose elements are difficult to enumerate or count, sampling is a widely used approach. For a query graph $q$ and a data graph $G$, let $\mathcal{M}$ denote the set of all isomorphic embeddings of $q$ in $G$. A sampling algorithm first defines a sample space $\Omega$ which is a superset of $\mathcal{M}$, and aims to approximate the ratio $\rho = \abs{\mathcal{M}} / \abs{\Omega}$.
Assume that (1) obtaining the cardinality of $\Omega$ is easier than that of $\mathcal{M}$, (2) $\Omega$ is amenable to random sampling, and (3) verification of whether a random sample $x \in \Omega$ is in $\mathcal{M}$ can be done efficiently. 
By obtaining uniform random samples from $\Omega$ and using the empirical ratio $\hat{\rho}$ which is the proportion of random samples lying in $\mathcal{M}$ as an estimator to $\rho$, the estimate $\abs{\Omega}\hat{\rho}$ is an unbiased consistent estimator \cite{StatisticalInference} for $\abs{\mathcal{M}}$. 

Our approach reduces the size of sample space $\Omega$ greatly by using filtering, while retaining all embeddings (i.e., keeping $\abs{\mathcal{M}}$ unchanged). That is, we increase the ratio $\rho = \abs{\mathcal{M}}/\abs{\Omega}$ by decreasing $\abs{\Omega}$, which makes sampling more accurate and more efficient. (Previous approaches take samples from the whole sample space $\Omega$, and thus the ratio $\rho$ is very small in many cases, thereby causing sampling failures.)

\begin{itemize}
    \item We build an auxiliary data structure \textit{Candidate Space} (CS) which contains candidates for vertex mapping and edge mapping. By employing novel safety conditions---Triangle Safety, Four-Cycle Safety, and Edge Bipartite Safety---and the refinement order called Promising First Candidate Filtering, we reduce the candidate edges by up to 80\% compared to the filterings of the state-of-the-art subgraph matching algorithms \cite{VEQ, VCSubgraphMatching}, resulting in a compact CS.
    \item We develop a \textit{candidate tree sampling} algorithm, which performs uniform sampling of spanning trees of the query graph in the compact CS. 
    Our tree sampling employs an adaptive strategy to determine the sample size using the Clopper-Pearson confidence interval, thereby achieving efficient and accurate estimation with rigorous probabilistic guarantees.
    \item We devise a \textit{stratified graph sampling} algorithm that obtains samples from diverse regions of the sample space, achieving outstanding accuracy on hard instances with a small sample size.
    Our graph sampling algorithm attains a worst-case optimal time complexity of $O(AGM(q))$ for the query graph $q$, which is a tight upper bound for $\abs{\mathcal{M}}$ proposed by Atserias, Grohe and Marx \cite{AGMBound, SSTE}, guaranteeing the same time complexity even when the sampling ratio is 100\%.
\end{itemize}
We demonstrate that \Fastest shows significant improvements in both accuracy and efficiency for subgraph cardinality estimation through extensive experiments on well-known real-world datasets. Specifically, \Fastest outperforms state-of-the-art sampling-based methods by up to two orders of magnitude and GNN-based methods by up to three orders of magnitude in terms of accuracy.

\paragraph{Organization} The rest of the paper is organized as follows. \Cref{sec:Prelim} introduces the definitions and the problem statement. \Cref{sec:Overview} outlines an overview of \Fastest. \Cref{sec:Filtering} discusses the candidate filtering algorithm. \Cref{sec:sampling} describes the candidate tree sampling algorithm. \Cref{subsec:CGSampling} explains the stratified graph sampling algorithm. \Cref{sec:experiments} presents the experimental results, and \Cref{sec:conclusion} concludes the paper. 
Proofs are provided in the appendix.

\section{Preliminaries}
\label{sec:Prelim}
\subsection{Problem Statement}

\begin{figure}[t]
    \includegraphics[width=0.85\linewidth]{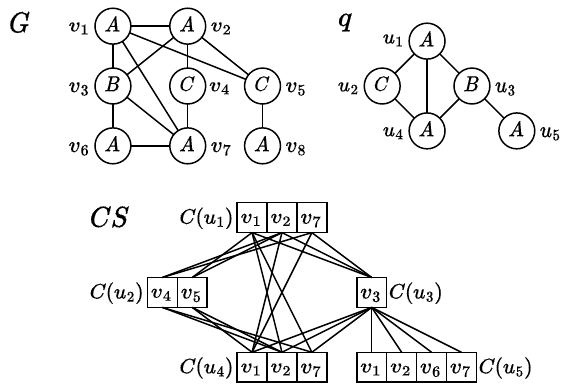}
    \caption{Data graph $G$, query graph $q$, and Candidate Space}
    \label{fig:initCS}
\end{figure}

In this paper, we mainly consider undirected and connected graphs with vertices labeled, while the techniques we propose can be extended to directed, disconnected, and edge-labeled graphs. We define a graph $g = (V_g, E_g, L_g)$ as a triplet of a set of vertices $V_g$, a set of edges $E_g$, and a labeling function $L_g : V_g \to \Sigma_V$ where $\Sigma_V$ is a set of possible labels.
Given query graph $q = (V_q, E_q, L_q)$ and data graph $G = (V_G, E_G, L_G)$, a vertex-mapping function $M : V_q \to V_G$ is called an \textit{(isomorphic) embedding} when three conditions are met: 
(1) $M$ is injective, i.e., $M(u) \neq M(u')$ for $u \neq u' \in V_q$. (2) $M$ preserves all adjacency relationships, i.e., $(M(u), M(u')) \in E_G$ for every $(u, u') \in E_q$, (3) $M$ preserves all vertex labels, i.e., $L_G(M(u)) = L_q(u)$ for every $u \in V_q$. 
We say that $q$ is \textit{subgraph isomorphic} to $G$ when such an embedding exists. 
For example, in \Cref{fig:initCS}, $\set{(u_1, v_1), (u_2, v_5), (u_3, v_3), (u_4, v_2), (u_5, v_6)}$ is an isomorphic embedding of $q$ in $G$, as it is injective and it preserves all vertex labels and adjacency relationships.
When a vertex-mapping function satisfies only (2) and (3), we call it a \textit{homomorphic embedding}.

\noindent\autoref{tab:notations} lists the notations that are frequently used in the paper.

\begin{table}[b]
    \caption{Frequently Used Notations}\label{tab:notations}
    \centering
    \resizebox{1.0\linewidth}{!}{
    \begin{tabular}{ll}
    \toprule
    Notation & Definition \\\midrule
    $q, G$ & Query graph and data graph \\ 
    $V_g, E_g, L_g$ & Vertices, edges, and labels of a graph $g$ \\ 
    $\mathcal{M}$ & Set of isomorphic embeddings of $q$ in $G$ \\ 
    
    $d_g(x)$ & Degree of a node $x$ in a graph $g$ \\ 
    $N_g(x)$ & Neighbors of $x$ in a graph $g$\\
    $d_g(x, l)$ & Number of vertices in $N_g(x)$ with label $l$ in a graph $g$\\ 
    
    $\delta_g$, $\Delta_g$ & Degeneracy \cite{Degeneracy} and maximum degree of a graph $g$ \\

    $T_q$ & A spanning tree of $q$ \\ 
    $C(u)$ & Set of candidate vertices for $u$ \\
    $C(u'\mid u,v)$ & Set of candidate neighbors of $(u, v)$ to $u'$ \\ 
    $E_{CS}(u, u')$ & Set of candidate edges for $(u, u')$ \\ 
    \bottomrule
    \end{tabular}
    }
\end{table}

\paragraph{Subgraph Matching and Counting} For query graph $q$ and data graph $G$, let $\mathcal{M}$ denote the set of all \textit{embeddings} of $q$ in $G$. The \textit{subgraph matching} problem is to find $\mathcal{M}$ exactly, while the \textit{subgraph counting} problem is to find $\abs{\mathcal{M}}$. It is clear that all subgraph matching algorithms can also be used for counting.

\paragraph{Problem Statement} Given a query graph $q$ and a data graph $G$, the \textit{subgraph cardinality estimation} problem is to approximate $\abs{\mathcal{M}}$, the cardinality of the set of all isomorphic embeddings of $q$ in $G$.

\subsection{Related Works}
\paragraph{Subgraph Matching and Counting}
In the recent decade, numerous solutions have been proposed for both exact and approximate subgraph matching and subgraph counting.

A wealth of results for subgraph matching \cite{TurboIso, DAF, VEQ, VF2, CFLMatch, VCSubgraphMatching, RISubgraphMatching, CECI} has been established based on Ullmann's backtracking framework \cite{Ullmann1976}. Recent works such as \cite{CFLMatch, DAF, VEQ, VCSubgraphMatching} share the filtering-backtracking approach. 
These works first build an auxiliary data structure such as Compact Path Index (CPI) \cite{CFLMatch}, Bigraph Index (BI) \cite{VCSubgraphMatching} and Candidate Space (CS) \cite{DAF, VEQ}, utilizing diverse filtering strategies such as extended DAG-graph DP using Neighbor Safety \cite{VEQ}. 
These algorithms utilizes such data structures to reduce search space for backtracking, which is further accelerated by strategic selection of matching orders such as Candidate-size order \cite{DAF}.
Interested readers can refer to a detailed survey of the aforementioned algorithms~\cite{RapidExperiment}.

Approximate solutions are split into three sub-classes: summariza\allowbreak tion-based, sampling-based, and machine learning-based approaches. 
Summarization methods, such as SumRDF \cite{SumRDF} and Characteristic Sets \cite{CSET}, construct a summary data structure of the data graph. Estimates for query graph $q$ are then processed efficiently by querying the summary structure, typically by decomposing $q$ into smaller substructures such as stars and aggregating the results computed for each substructure. However, these methods may produce highly inaccurate results \cite{GCare}.

Sampling algorithms have also been widely employed for subgraph counting \cite{DPColorPath, TETRIS, TuranShadow}.
WanderJoin \cite{WanderJoin} performs random walks over relations and uses Horvitz-Thompson estimation, with the calculated weights of each walk. 
JSub \cite{JSub} obtains independent uniform samples from join results by computing upper bounds for intermediate join sizes.
Alley \cite{Alley} proposes the random walk with intersection and it also uses Tangled Pattern Index (TPI), which is a summary structure obtained by mining difficult patterns.

Machine learning-based solutions, particularly those utilizing Graph Neural Networks (GNN), have recently begun to emerge. Works like LSS \cite{LSS} and NeurSC \cite{NeurSC} have explored the potential of GNNs for subgraph cardinality estimation. While machine learning enables highly efficient query processing, these solutions does not provide an unbiased estimation.

\paragraph{Counting Special Classes of Subgraphs} Extensive research has been conducted on subgraph cardinality estimation of special classes. TETRIS \cite{TETRIS} develops a sublinear approximate algorithm for counting triangles under the random walk access model. 
TuranShadow \cite{TuranShadow} is a randomized sampling algorithm for estimating clique counts based on Turan's theorem. 
DPColorPath \cite{DPColorPath} suggests a uniform sampling algorithm combined with greedy coloring, which also estimates the number of cliques.

\paragraph{AGM Bound and Worst-Case Optimality}
The AGM bound~\cite{AGMBound} is a tight upper bound for the number of embeddings of a query graph $q$ in a data graph $G$.
An algorithm for subgraph matching or cardinality estimation is referred to as \textit{worst-case optimal} if it guarantees $O(AGM(q))$ time complexity~\cite{WCOSubgraphQueryProcessing, Alley}.
SSTE \cite{SSTE} developed an edge sampling algorithm for estimating arbitrary-sized subgraph counts in $\tilde{O}(AGM(q) / OUT(q))$ time, where $OUT(q)$ refers to the number of embeddings. While it provides strong theoretical guarantees, its performance on real-world graphs is not impressive~\cite{GJSampler}.

\section{Overview of Algorithm}
\label{sec:Overview}
\setlength{\textfloatsep}{0pt}
\begin{algorithm}[t]
    \caption{\textsc{Overview of} \Fastest}
    \label{alg:overview}
    \rm\sffamily
    \KwIn{A data graph $G$, a query graph $q$}
    \KwOut{Estimated number of embeddings of $q$ in $G$}
    
     CS $\gets$ BuildCompactCandidateSpace($G, q$)\;

    (estimate, \#successes) $\gets$ CandidateTreeSampling(CS, $q$)\; 

    \lIf{\rm \sffamily \#successes is large enough}{\Return{\rm\sffamily estimate}}
    \lElse{\Return{\rm \textsf{CandidateGraphSampling}(\textsf{CS}, \(q\), \#successes)}}
\end{algorithm}
\setlength{\textfloatsep}{3pt}

\paragraph{Filtering-Sampling Approach} Since the problem of deciding whether $q$ is subgraph isomorphic to $G$ is NP-hard, exact enumeration or counting of embeddings is often computationally infeasible. 
To approximate the cardinality of a set $\mathcal{M}$ whose elements are difficult to enumerate or count, sampling is a widely used approach. A sampling algorithm first defines a sample space $\Omega$ which is a superset of $\mathcal{M}$, and aims to approximate the ratio $\rho = \abs{\mathcal{M}} / \abs{\Omega}$.

Our approach reduces the size of sample space $\Omega$ greatly, while retaining all embeddings. We develop a filtering algorithm to reduce the sample space by removing unnecessary vertices and edges in finding all embeddings.

\paragraph{Filtering} 
We build an auxiliary data structure \textit{Candidate Space} (CS) containing candidates for vertices and edges. To build a compact CS, we perform iterative refinements for each vertex $u \in V_q$. 
The refinement step involves checking a set of necessary safety conditions that candidates must satisfy, which enables us to remove invalid candidates.
\begin{itemize}
    \item \textit{Generic Framework}. We propose the generalized framework for candidate filtering, where different conditions and refinement strategies can be implemented. (\Cref{subsec:filteringframework})
    \item \textit{Safety Conditions}. We propose novel safety conditions: \textit{Triangle Safety}, \textit{Four-Cycle Safety}, and \textit{Edge Bipartite Safety}. These conditions effectively reduce the CS by better utilizing local substructures. (\Cref{subsec:safetyconditions})
    \item \textit{Refinement Order}. We suggest \textit{Promising First Candidate Filtering} as the refinement order, where in each iteration a candidate set most likely to be reduced is chosen to be refined. (\Cref{subsec:PFCF})
    
\end{itemize}
The combined filtering algorithm we develop is much stronger than those in \cite{DAF, VEQ, VCSubgraphMatching} in terms of filtering power, resulting in a compact CS.

\paragraph{Sampling} We develop a sampling algorithm utilizing the compact CS for accurate and efficient subgraph cardinality estimation.
\begin{itemize}
    \item \textit{Candidate Tree Sampling}. We choose a spanning tree $T_q$ of $q$, and define the sample space $\Omega$ as the set of homomorphisms (called candidate trees) of $T_q$ in the compact CS. 
    We count the number of candidate trees and sample them uniformly at random. ~(\Cref{sec:sampling})
    \item \textit{Stratified Graph Sampling}. We develop a graph sampling algorithm to handle hard cases in which obtaining accurate results via tree sampling is not possible in a reasonable time. 
    We propose a novel sampling method inspired by stratification to sample from diverse regions of the sample space, achieving worst-case optimality. (\Cref{subsec:CGSampling}) 
\end{itemize}
\noindent The outline of our algorithm \Fastest is shown in \Cref{alg:overview}.

\section{Candidate Filtering}
\label{sec:Filtering}

Candidate Space (CS) is an auxiliary data structure built for efficient subgraph matching, proposed by DAF \cite{DAF} and extended by VEQ \cite{VEQ}. 
Here we extend the definition of CS by explicitly maintaining candidate edges, which is necessary for better filtering.
\subsection{Candidate Space}
\begin{definition}[Candidate Space]
    For a vertex $u \in V_q$ and its neighbor $u' \in N_q(u)$, we define the following terms.
    \begin{itemize}
        \item \textit{Candidate set} $C(u)$: a subset of $V_G$ such that if there exists an embedding $M$ of $q$ in $G$ that maps $u$ to $v$, $v$ must be included in $C(u)$.
        \item \textit{Candidate edges} $E_{CS}(u, u')$: a subset of $E_G$ such that if there exists an embedding $M$ of $q$ in $G$ that maps $u$ to $v$ and $u'$ to $v'$, $(v, v')$ must be included in $E_{CS}(u, u')$.
    \end{itemize} 
    The \textit{Candidate Space} is a collection of candidate sets and candidate edges for query graph $q$ and data graph $G$.
\end{definition}

By explicitly maintaining candidate edges, we can remove an edge $(v,v')$ from candidate edges if it does not belong to any embeddings, even though vertices $v$ and $v'$ remain in the Candidate Space.
For $v \in C(u)$ and $u' \in N_q(u)$, we define the \textit{candidate neighbors} of $(u, v)$ to $u'$ as the vertex candidates of $u'$ when $u$ is mapped to $v$, i.e., $C(u' \mid u, v) = \set{v' \mid (v, v') \in E_{CS}(u, u')}$.

We build the initial CS as follows. For each vertex $u$, initial candidate set $C_{init}(u)$ is defined as the set of vertices $v \in V_G$ with the following two conditions: (1) $v$ has the same label as $u$; (2) for each label $l \in \Sigma$, $d_G(v, l) \geq d_q(u, l)$.
We initialize candidate edges as $E_{CS}(u, u') = \set{(v, v') \in E_G \mid v \in C(u) \text{ and}\  v' \in C(u')}$. The initial candidate space can be constructed in $O(\abs{E_q} \abs{E_G})$ time.

\begin{example}
    \Cref{fig:initCS} depicts an example of the initial CS given query graph $q$ and data graph $G$. 
    Among the vertex pairs with matching labels, $v_6$ is not in $C(u_1)$ since $d_G(v_6, C) < d_q(u_1, C)$. Other absences can also be seen analogously. 
    In the CS, $C(u_4 \mid u_1, v_1)$, the candidate neighbors of $(u_1, v_1)$ to $u_4$, is $\set{v_2, v_7}$ as they are neighbors of $v_1$ which are in $C(u_4)$.
\end{example}

By definition, all embeddings of $q$ in $G$ are preserved in CS, as all possible vertices and edges constituting an embedding are present. We call this property of CS as \textit{complete}.
In our sampling algorithm later to be described, reducing the size of a candidate space translates to a smaller sample space. Therefore, the key is to refine the candidate space as much as possible while retaining completeness.
Note that CS can contain $O(\abs{V_q} \abs{V_G})$ vertices and $O(\abs{E_q}\abs{E_G})$ edges in the worst case.

\label{subsec:filteringframework}
\begin{algorithm}[t]
    \caption{\textsc{Candidate Space Refinement}}
    \rm\sffamily
    \label{alg:Refine}
    \SetKwProg{Fn}{Function}{:}{}
    \While{\rm\sffamily refinement is not finished} {
        $u$ $\gets$ a vertex from query graph\;
        \label{line:choosevertex}
        \ForEach{$v \in C(u)$}{
            \If{\rm\sffamily VertexSafety$(u, v)$ violated} {
                Remove $v$ from $C(u)$\;
                continue\;
            }
            \ForEach{$u' \in N_q(u)$} {
                \ForEach{$v' \in C(u' \mid u, v)$} {
                    \If{\rm\sffamily EdgeSafety$((u, u'), (v, v'))$ violated} {
                        Remove $v'$ from $C(u' \mid u, v)$\;
                    }
                }
            }
            \If{\rm\sffamily $C(u' \mid u, v)$ is empty for any $u'\in N_q(u)$} {
            \label{line:EmptyCandidateNeighbor}
                Remove $v$ from $C(u)$\;
            }
        }
    }
\end{algorithm}
\paragraph{Generic Filtering Framework}
We start the refinement from the initial CS. For a predefined set of safety conditions, the generic filtering framework can be defined as \Cref{alg:Refine}.
The framework consists of several design choices; a set of \textit{safety conditions} to determine validity of candidate vertices or edges, a method to decide the \textit{refinement order}, i.e., the order of query vertices whose candidate set is refined, and the \textit{stopping criteria} by which we finish the refinement. 

\subsection{Safety Conditions}
\label{subsec:safetyconditions}
It is necessary that $C(u'\mid u, v)$ is nonempty for $v$ to be a candidate for $u$ for each neighbor $u'$ of $u$. Together with this condition, \VEQ~\cite{VEQ} suggested the use of \textit{safety condition} that can further filter invalid candidates. For any necessary condition $h(u, v)$ that must be satisfied whenever $M(u) = v$ for some embedding $M$, in any step, we can remove candidate vertex $v$ from $C(u)$ when $h(u, v)$ is false. 
\VEQ considers \textit{Neighbor Safety} for $h$, that can filter candidate vertices if it lacks sufficient candidate neighbors for some label. 
\GQL \cite{GraphQL} and \VC \cite{VCSubgraphMatching} employ a stronger condition using bipartite matching.

As we explicitly maintain candidate edges, any necessary condition $g((u, u'), (v, v'))$ for edges which must be satisfied whenever $M(u) = v$ and $M(u') = v'$ for an embedding $M$ can be used to remove invalid candidate edges. 

In this paper, we propose novel safety conditions for stronger filtering. 
We propose the Triangle Safety and Four-Cycle Safety which consider cyclic substructures. 
We extend the ESIC condition suggested by \VC~\cite{VCSubgraphMatching} to Edge Bipartite Safety, increasing the filtering power while maintaining the same time complexity.

\paragraph{Triangle Safety}
Filtering methods in \cite{DAF, VEQ, CFLMatch, VCSubgraphMatching} did not utilize cyclic substructures, i.e., triangles and 4-cycles in the query graph.
While finding all cycles is itself an NP-hard problem, using small cycles can greatly increase the filtering power without taking excessive time.

\begin{definition}
\label{def:localtriangles}
    For an edge $(a, b) \in E_g$, we define $L^3_g(a, b)$ as the set of vertices $c \in V_g$ such that $(a, b), (b, c), (c, a) \in E_g$, i.e., they form a triangle with the given edge $(a, b)$.
\end{definition}
\noindent We define the condition \textit{Triangle Safety} as follows. 
\begin{restatable}{definition}{DefTriangleSafety}
\label{def:trianglesafety}
A data edge $(v, v')$ for a query edge $(u, u')$ is \textit{Triangle Safe} when 
\begin{itemize}
\item $\abs{L^3_q(u, u')} \leq \abs{L^3_G(v, v')}$, and
\item for each $u^* \in L^3_q(u, u')$,
there exists $v^* \in L^3_G(v, v')$ consisting of valid vertices and edges, i.e., for each $u^*$, 
there exists $v^* \in C(u^*\mid u,v) \cap C(u^*\mid u',v')$. 
\end{itemize}
\end{restatable}

\begin{example}
    \Cref{fig:triangle-safety} shows an example of refinement by Triangle Safety. Considering the query edge $(u_2, u_4)$ and data edge $(v_4, v_2)$, Triangle Safety evaluates edge safety conditions where a query triangle $(u_2, u_4, u_1)$ can be matched to a valid triangle in CS. Since $C(u_1 \mid u_2, v_4) \cap C(u_1 \mid u_4, v_2) $ is empty, edge $(v_4, v_2)$ is removed from candidate edges for $(u_2, u_4)$ by \Cref{def:trianglesafety}. As $(v_4, v_7)$ for $(u_2, u_4)$ is also removed by Triangle Safety, $C(u_4 \mid u_2, v_4)$ becomes empty and therefore $v_4$ can be removed from $C(u_2)$.
\end{example}

\paragraph{Four-Cycle Safety} 
A similar work can be done with 4-cycles in the query graph and the data graph.

\begin{definition}
\label{def:local4cycle}
    For an edge $(a, b) \in E_g$, we define $L^4_g(a, b)$ as the set of edges $(c,d) \in E_g$ such that $(a, b),\allowbreak (b, c),\allowbreak (c, d),\allowbreak (d,a) \in E_g$, i.e., they form a 4-cycle with the given edge $(a, b)$.
\end{definition}
\noindent We define the condition \textit{Four-Cycle Safety} as follows. 
\begin{definition}
\label{def:4cyclesafety}
    A data edge $(v, v')$ for a query edge $(u, u')$ is \textit{Four-Cycle Safe} when 
    \begin{itemize}
        \item $\abs{L^4_q(u, u')} \leq \abs{L^4_G(v, v')}$, and
        \item for each $(u^*, \hat{u}) \in L^4_q(u, u')$,
        there exists $(v^*,\hat{v}) \in L^4_G(v, v')$ consisting of valid vertices and edges, i.e., for each $(u^*, \hat{u})$, 
        there exists $v^*\in C(u^*\mid u',v')$ and $\hat{v} \in C(\hat{u}\mid u,v)$ such that $(v^*,\hat{v})\in E_{CS}(u^*,\hat{u})$. 
    \end{itemize}
\end{definition}

To compute Triangle Safety and Four-Cycle Safety, we index all triangles and 4-cycles in advance. 
When the numbers of triangles and 4-cycles in the data graph are substantial, indexing them can result in a significant overhead. Hence, if the numbers of triangles and 4-cycles exceed some thresholds $k_1$ and $k_2$, respectively, we disable respective safety condition to avoid excessive overhead of indexing. For the experiments, we use $k_1=k_2=10^{9}$.

\begin{figure}[t]
    \centering
    {
    \includegraphics[width=1.0\linewidth]{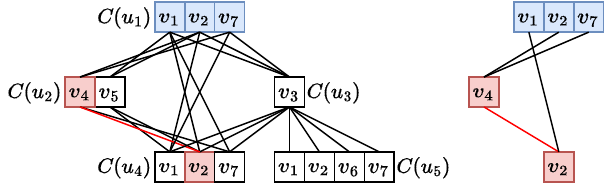}
    \subcaption{Example of Triangle Safety}
    \label{fig:triangle-safety}
    \vspace{0.6em}
    
    \includegraphics[width=1.0\linewidth]{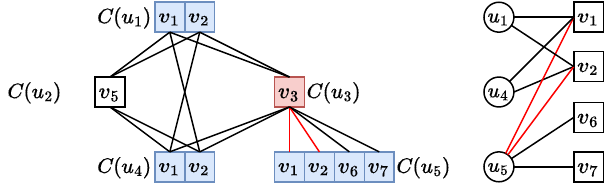}
    \subcaption{Example of Edge Bipartite Safety}
    \label{fig:edge-bipartite-safety}
    }
    
    \caption{Running Example of Candidate Filtering}
    \label{fig:filtering}
\end{figure}

\paragraph{Edge Bipartite Safety}
To develop a stronger safety condition, \GQL \cite{GraphQL} and \VC \cite{VCSubgraphMatching} propose the use of a bipartite graph $B(u, v)$ for candidate neighbors of $v \in C(u), u \in V_q$.
$B(u, v)$ consists of left vertices $V_L = N_q(u)$, and right vertices $V_R = N_G(v)$.
The edges of $B(u, v)$ connect $u'$ and $v'$ if and only if $v' \in C(u'\mid u, v)$, i.e., when a candidate edge exists between $v\in C(u)$ and $v'\in C(u')$.

\begin{example}
    \Cref{fig:edge-bipartite-safety} depicts $B(u_3, v_3)$ during the refinement of $C(u_3)$, given the query graph and the data graph in \Cref{fig:initCS} and the CS in the left-hand side.
    Neighbors $u_1, u_4, u_5$ of $u_3$ in the query graph become left vertices $V_L$ of $B(u_3, v_3)$, and neighbors $v_1, v_2, v_6, v_7$ of $v_3$ in the data graph become right vertices $V_R$.
\end{example}

It is then readily observed that for $v$ to be a valid candidate for $u$, $B(u, v)$ should have a bipartite matching whose size is equal to the set of left vertices.
We can determine whether $B(u, v)$ has such a bipartite matching in $O(d_q(u)^2 d_G(v))$ time utilizing \textsf{Ford-Fulkerson} \cite{FordFulkerson} algorithm, as there are at most $d_q(u) d_G(v)$ edges and the size of the maximum matching is at most $d_q(u)$.

We further extend the use of the bipartite graph to the fine-grained filtering of candidate edges. 
Given the bipartite graph $B(u, v)$, the existence of a maximum bipartite matching of size $d_q(u)$ does not guarantee that for every edge there exists a maximum matching containing the edge.

\begin{definition}
    An edge is called \textit{maximally matchable} if it is included in some maximum matching.
\end{definition}
When an edge $(u', v') \in E_{B(u, v)}$ is not maximally matchable, $v'$ should be removed from $C(u'\mid u,v)$. 
For bipartite graphs, it is known that finding all maximally matchable edges can be done in linear time when a maximum matching is given~\cite{MaximallyMatchableEdges}. Therefore, after we find the maximum bipartite matching on $B(u, v)$, we verify whether each edge is maximally matchable. This takes $O(d_q(u) d_G(v))$ time which is negligible compared to finding maximum bipartite matching.

\begin{example}
    In $B(u_3, v_3)$ depicted in \Cref{fig:edge-bipartite-safety}, a maximum matching $\set{(u_1, v_1), (u_4, v_2), (u_5, v_6)}$ of size 3 can be found, thus $v_3$ is a valid candidate for $u_3$. However, since there is no maximum matching including edge $(u_5, v_1)$ or $(u_5, v_2)$, these edges are not maximally matchable. Therefore $v_1$ and $v_2$ are removed from $C(u_5 \mid u_3, v_3)$.
\end{example}

\subsection{Promising-First Candidate Filtering}
\label{subsec:PFCF}

\paragraph{Refinement Priority}
In each refinement step, we select a query vertex to be refined (\Cref{alg:Refine}, \Cref{line:choosevertex}).
Let $C_t(u)$ be the candidate set of $u$ after the $t$-th refinement step finished, with $C_0(u)$ denoting the candidate set in the initial CS. 
Let $u^1, u^2, \dots, u^t$ be the sequence of vertices that have already been refined, with the current refinement step being $t+1$. 

As $u^{t+1}$, we choose the query vertex $u$ such that $u$ has the lowest $penalty_t(u)$ (i.e., most promising), defined as
\begin{equation*}
    \label{eqn:priority}
    penalty_{t}(u) = \begin{cases}
        1 & \text{ if } u = u^t\\
        penalty_{t-1}(u) \times \frac{\abs{C_{t}(u^t)}}{\abs{C_{t-1}(u^t)}}  & \text{ if } u \in N_q(u^t)\\ 
        penalty_{t-1}(u)  & \text{ otherwise.}
    \end{cases}
\end{equation*}
We assign all vertices an initial penalty of $penalty_{0}(u) = \phi$ for some constant $\phi$ in $[0, 1]$.
After refining $u^t$, $penalty_t(u^t)$ is increased to 1.
For $u \in N_q(u^t)$, we reduce $penalty_t(u)$ by multiplying $\abs{C_t(u^t)}/\abs{C_{t-1}(u^t)}$ (the ratio by which $C(u^t)$ is reduced), 
as the size of $C(u)$ is also likely to decrease when the size of $C(u^t)$ decreases. 
If $u$ is neither $u^t$ itself nor a neighbor of $u^t$, its penalty remains unchanged.
Note that the penalty always lies within the range $[0, 1]$.

\begin{example}
    For given $q$ and the initial CS (\Cref{fig:initCS}), \Cref{fig:filtering} depicts the filtering process. Each vertex in $V_q$ is given an initial penalty of $2/3$. Among the vertices, we break the ties arbitrarily. In \Cref{fig:triangle-safety}, $u_2$ is chosen for the refinement. With enhanced filtering power of Triangle Safety, $v_4$ is removed from $C(u_2)$, reducing penalties of $u_1$ and $u_4$ to $1/3$ (multiplying $1/2$, the ratio by which $C(u_2)$ is reduced) and resetting the penalty of $u_2$ to $1$. The penalties of $u_3$ and $u_5$ remain unchanged as they are not in $N_q(u_2)$. 
\end{example}

\paragraph{Stopping Criteria} We empirically observe that the efficacy of refinements stagnates after some iterations. Therefore, we define two stopping criteria for the refinement, and stop whenever either one of the two conditions triggers.
\begin{enumerate}
    \item When the minimum penalty among query vertices become higher than some predefined constant $\tau$. This explicitly guides the algorithm to stop when the efficacy of further refinements are expected to be small and not worth the overhead. 
    \item When the sum of degrees of vertices chosen to be refined, including repeated selections, becomes higher than $R \abs{E_q}$ for some predefined constant $R$. This criterion is necessary to bound the worst case complexity, and practically guides the algorithm to stop when the refinement takes too much time.
\end{enumerate}

\begin{restatable}{theorem}{ThmTotalComplexity}
\label{thm:total_complexity}
    With stated stopping criteria, various safety conditions lead to the following time complexity bounds for the filtering.
    \begin{itemize}
        \item Edge Bipartite Safety : $O(\abs{E_q}\abs{E_G}\Delta_q)$ time
        \item Triangle Safety : $O(\abs{E_q}\abs{E_G}\Delta_q \delta_G)$ time
        \item 4-Cycle Safety : $O(\abs{E_q} \abs{E_G} \Delta_q^2\Delta_G \delta_G)$ time
    \end{itemize}
\end{restatable}

Empirically, we choose $\tau = 0.9$ and $R = 5$ for the experiments. We observe that most of the runs were stopped by the first criterion~(penalty).

\section{Candidate Tree Sampling}
 \label{sec:sampling}
We present a sampling algorithm to select trees within the Candidate Space. Although there have been related works employing tree sampling for similar problems such as join cardinality estimation \cite{JSub}, these algorithms may face challenges in accuracy and efficiency since the number of trees is much larger than the actual number of embeddings. 
In contrast, our algorithm greatly benefits from the highly effective Candidate Space, which enables us to filter out numerous trees that cannot be extended to embeddings.

Our algorithm strategically chooses the spanning tree of the query graph to minimize the sample space. We prove that our method is optimal under an additional assumption, and it performs remarkably well in practice even when such an assumption is not met. Furthermore, we adaptively determine the sample size during the sampling process to provide a probabilistic bound. In most cases, our algorithm achieves the bound with a reasonably small sample size, ensuring both accuracy and efficiency in practice.

\subsection{Candidate Trees}
\label{subsec:candidatetrees}
\begin{definition}[Candidate Tree]
\label{def:candidatetree}   
    Let $T$ be a subgraph of $q$ which is a tree rooted at $u_r \in V_T$. We define a \textit{Candidate Tree} for $T$ as a vertex mapping $s : V_T \to V_G$ such that
    \begin{itemize}
        \item for the root node $u_r$ of $T$, $s(u_r) \in C(u_r)$, and
        \item for a node $x$ and its parent node $p \in V_T$, $s(x) \in C(x\mid p, s(p))$.
    \end{itemize}
\end{definition}
\noindent We note that a candidate tree can be seen as a homomorphism of $T$ in the compact CS.

\begin{example}
    In \autoref{fig:tree}, for given $q$ and compact CS, we first build a spanning tree $T_q$ of $q$ by removing $(u_2, u_4)$ and $(u_3, u_4)$ edges. The mapping $\set{(u_1, v_1), (u_2, v_3), (u_3, v_5), (u_4, v_7), (u_5, v_3})$ is a candidate tree for $T_q$. Note that the candidate tree is not necessarily an injective mapping.
\end{example}

Let $T_q$ be any spanning tree of $q$. When needed, we consider $T_q$ as a directed graph with edge directions assigned from root to leaf. 
We say that a candidate tree $t$ is a \textit{counterpart} of an embedding if the vertex mapping of $t$ is the same as the vertex mapping $M$ of the embedding. Since the CS is complete, the set of candidate trees for $T_q$ contains the counterparts of all the embeddings of $q$.
Therefore, we choose a spanning tree $T_q$ and use the set of candidate trees for $T_q$ in the compact CS as the sample space $\Omega$. We develop an algorithm to perform exact counting of candidate trees using dynamic programming, and sample a candidate tree from $\Omega$ uniformly at~random.

We note that the sample space can be further reduced if the definition of candidate trees is extended to only allow injective~mappings. However, this extension is infeasible, as exact counting of isomorphic embeddings is NP-hard for tree queries \cite{SubIsoHardness}.

\begin{algorithm}[t]
    \caption{\textsc{CandidateTreeSampling(CS, $q$)}}
    \label{alg:treesampling}
    \rm\sffamily
    \KwIn{The compact CS, a query graph $q$}
    \KwOut{Estimated number of embeddings and the number of successful trials in tree sampling}
    
    $T_q$ $\gets$ GetQueryTree(CS, $q$)\;
    
    $D$ $\gets$ CountCandidateTrees(CS, $T_q$)\;
    
    \#successes $\gets$ 0\;

    \#trials $\gets$ 0\;

    \While{\rm \sffamily sampling termination condition is not reached} {
        s $\gets$ GetSampleTree(CS, $T_q, D$)\;
        \#trials $\gets$ \#trials + 1\;
    
        \If{\rm\sffamily CheckEmbedding(CS, s, $q$)}{
            \#successes $\gets$ \#successes + 1\;
        }
    }
    $u_r$ $\gets$ root of $T_q$\;
    \vspace{0.3em} 
    estimate $\gets$ $\displaystyle \mathsf{\frac{\#successes}{\#trials}} \sum_{v \in C(u_r)} D(u_r, v) $ \; 
    
     \Return{\rm \sffamily (estimate, \#successes)\;}
\end{algorithm}
\vspace{-2mm}

\subsection{Sampling Candidate Trees}
\label{subsec:TreeSampling}

\paragraph{Choosing the Spanning Tree}
It is best to choose the spanning tree $T_q$ that the number of candidate trees is minimized to reduce the size of the sample space. However, exact minimization by generating all possible spanning trees and counting the number of candidate trees is not viable as the number of spanning trees for $q$ can be superexponentially large.

Thus, instead of exact minimization, we introduce a \textit{density heuristic} to find a spanning tree with a small number of candidate trees. 
Since the number of candidate trees generally increases as the number of candidate edges increases, using query edges having dense candidate edges is undesirable.
To avoid such cases, we assign the weight of the edge $(u, u') \in E_q$ as 
\begin{equation}
    {density}(u, u') = \frac{\abs{E_{CS}(u, u')}}{\abs{C(u)}\abs{C(u')}}
\end{equation}
and find a spanning tree that minimizes the sum of the logarithm of densities by \textsf{Prim's Algorithm} \cite{Prim}, which is equivalent to minimizing the product of densities of the selected edges.

This strategy of minimizing the product of density is optimal with the following assumption:
for all $(v, v') \in C(u) \times C(u')$, the event that edge $(v, v')$ is a candidate for $(u, u')$ has the probability ${density}(u, u')$, and it is independent from the events that other vertex pairs are candidate edges.

\begin{restatable}{theorem}{ThmDensityHeuristicOptimality}
\label{thm:density_heuristic_optimality}
Under the above assumption, the expected number of candidate trees is minimized by the density heuristic.
\end{restatable}

We demonstrate in \Cref{subsec:ablationstudy} that our strategy outperforms other intuitively appealing alternatives.

\paragraph{Counting Candidate Trees}
We develop a dynamic programming algorithm to obtain the exact count of candidate trees for a given $T_q$. Let $T_u$ be a subtree of $T_q$ rooted at $u \in V_q$ and $D(u, v)$ be the number of candidate trees for $T_u$ with $u$ mapped to $v$. The total number of candidate trees for $T_u$ can be counted as $\sum_{v \in C(u)} D(u, v)$. 

For a leaf node $u \in V_q$ and $v \in C(u)$, it is clear that $D(u, v) = 1$. 
For a non-leaf node $u \in V_q$ and $v \in C(u)$, a candidate tree for $T_u$ consists of candidate trees for $T_{u_c}$ for each child $u_c$ of $u$,
and the candidate for each $u_c$ is chosen from $C(u_c\mid u,v)$. Hence, we have
\begin{equation}
\label{eq:numtree-recurrence}
    D(u, v) = \prod_{u_c \in \text{children of } u} \ \sum_{v_c \in C(u_c\mid u,v)} D(u_c, v_c)
\end{equation}
which can be computed in $O(\abs{E_q}\abs{E_G})$ time, employing a bottom-up dynamic programming approach similar to \textsf{JSub} \cite{JSub} and \DAF~\cite{DAF}.

\begin{figure}[t]
    \centering
    \begin{subfigure}[t]{0.48\linewidth}
    \centering
    \includegraphics[width=\linewidth]{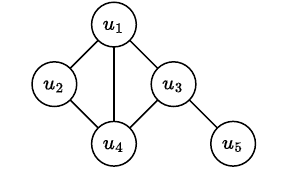}
    \subcaption{Query graph $q$}
    \label{subfig:CTSa}
    \end{subfigure}
    \begin{subfigure}[t]{0.48\linewidth}
    \centering
    \includegraphics[width=\linewidth]{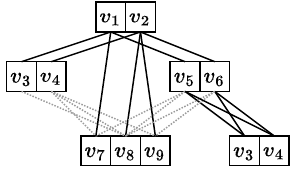}
    \subcaption{Compact CS, Spanning tree $T_q$}
    \label{subfig:CTSb}
    \end{subfigure}
    
    \vspace{0.2em}
    
    \begin{subfigure}[t]{0.48\linewidth}
    \centering
    \includegraphics[width=\linewidth]{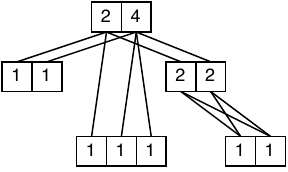}
    \subcaption{Counting Candidate Trees}
    \label{subfig:CTSc}
    \end{subfigure}
    \begin{subfigure}[t]{0.48\linewidth}
    \centering
    \includegraphics[width=\linewidth]{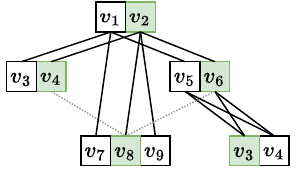}
    \subcaption{Success}
    \label{subfig:CTSd}
    \end{subfigure}
    \vspace{0.2em}
    
    \begin{subfigure}[t]{0.48\linewidth}
    \centering
    \includegraphics[width=\linewidth]{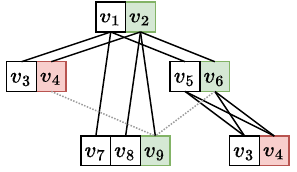}
    \subcaption{Failure (Injectivity)}
    \label{subfig:CTSe}
    \end{subfigure}
    \begin{subfigure}[t]{0.48\linewidth}
    \centering
    \includegraphics[width=\linewidth]{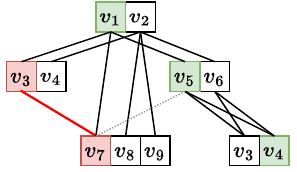}
    \subcaption{Failure (Non-tree edge)}
    \label{subfig:CTSf}
    \end{subfigure}
    \caption{New Running Example - Candidate Tree Sampling}
    \label{fig:tree}
    \vspace{-1mm}
\end{figure}

\paragraph{Uniform Sampling of Candidate Trees}
Based on the candidate tree counts computed as above, we develop a sampling algorithm that returns each candidate tree uniformly at random. Let $s$ denote a sample candidate tree for $T_q$. Recall that the candidate tree is defined as a vertex mapping.
For the root vertex $u_r$, we sample $v$ from $C(u_r)$ with weights proportional to $D(u_r, v)$. 
For $u \neq u_r$, let $u^p$ be the parent of $u$ in $T_q$.
Given $v^p = s(u^p)$, we iteratively sample $s(u)$ for $u$ with BFS traversal of $T_q$. We choose $s(u)$ from candidate neighbors $C(u\mid u^p, v^p)$ with weights proportional to $D(u, v)$. 

\noindent The detailed implementation is presented in \Cref{alg:RandomTree}.

\begin{algorithm}[t]
\rm\sffamily
    \caption{\textsc{GetSampleTree}(\textsf{CS}, $T_q$, $D$)}
    \label{alg:RandomTree}
    \KwIn{The compact \textsf{CS}, a spanning tree $T_q$ of  $q$, the computed weights $D$}
    \KwOut{A sampled candidate tree, uniformly at random}
    
    $u_r$ $\gets$ root node of $T_q$\;  
    \textsf{Sample}[$u_r$] $\gets$ draw $v$ from $C(u_r)$ at random with probability proportional to $D(u_r, v)$\;
    
    \ForEach{\rm\sffamily $u$ in BFS traversal of $T_q$} {
        let $u^p$ be the parent of $u$, and $v^p$ be Sample[$u^p$]\;
        \textsf{Sample}[$u$] $\gets$ draw $v$ from $C(u\mid u^p,v^p)$ at random with probability proportional to $D(u, v)$\;
    }
    \Return{\rm\sffamily \textsf{Sample}}
\end{algorithm}

\begin{example}
    \Cref{fig:tree} demonstrates the process of counting and sampling candidate trees. For the query graph $q$, the compact CS is built as in \Cref{subfig:CTSb}. 
    The edges $(u_2, u_4)$ and $(u_3, u_4)$ are not included to attain $T_q$ as they have larger densities of $2/3$. \Cref{subfig:CTSc} shows the value of resulting $D(u, v)$ after counting candidate trees.

    For the sampling, we start with choosing a root with weight $2 : 4$ for $v_1$ and $v_2$, respectively. For example in \Cref{subfig:CTSd}, when $v_2$ is chosen, $v_4$ and $v_6$ are the only candidate neighbors of $(u_1, v_2)$ for $u_2$ and $u_3$ respectively, thus they are sampled with probability 1.
    For $u_4$, we take a random sample from $C(u_4 \mid u_1, v_2)$. Analogously, $s(u_5)$ is chosen among $C(u_5 \mid u_3, v_6)$. As the resulting mapping in \Cref{subfig:CTSd} is injective, we check whether non-tree edges $(u_2, u_4)$ and $(u_3, u_4)$ are valid. Since $(v_4, v_8)$ and $(v_6, v_8)$ are both in the candidate edges, \Cref{subfig:CTSd} is a success.

    When $v_4$ is chosen for both $u_2$ and $u_5$ (\Cref{subfig:CTSe}), the resulting sample is an instance of failure since the mapping is not injective. If $s(u_2) = v_3$ and $s(u_4) = v_7$, the candidate for query edge $(u_2, u_4)$ is $(v_3, v_7)$. 
    However, $v_7 \not\in C(u_4\mid u_2, v_3)$, and thus the instance in \Cref{subfig:CTSf} is also a failure.
    Since one success was found among three, and there are 6 candidate trees, we return $2$ as the estimate. In reality there are two embeddings $(v_2, v_4, v_6, v_8, v_3)$ and $(v_2, v_4, v_6, v_9, v_3)$.

\end{example}

\begin{restatable}[Uniform Sampling]{theorem}{ThmTreeUniform}
\label{thm:tree_uniform}
    \Cref{alg:RandomTree} samples \textit{candidate tree} for $T_q$ uniformly at random.
\end{restatable}

\paragraph{Sample Size} 
If we perform $\#t$ trials and obtain $\#s$ successes, the sample success ratio $\hat{\rho}=\#s/\#t$ is an estimate of the true proportion $\rho$.
Let $\alpha$ denote the acceptable failure probability and $c$ denote the tolerable relative error.
To achieve the probability guarantee $\Prob{c^{-1} \rho \leq \hat{\rho} \leq c \rho} \geq 1 - \alpha$, we utilize the Clopper-Pearson interval~\cite{ClopperPearson, IntervalEstimation}, which is the interval $(L_\alpha(\#t,\#s), U_\alpha(\#t,\#s))$ such that
\begin{equation}
    \Prob{(L_\alpha(\#t,\#s) \le \rho \le U_\alpha(\#t,\#s))}\ge 1-\alpha,
\end{equation}
where $L_\alpha(\#t,\#s)$ and $U_\alpha(\#t,\#s)$ are the values that can be computed from $\alpha$, $\#t$, and $\#s$ (e.g., using the Boost library \cite{BoostLibrary}).
We adaptively determine the sample size by taking samples until $c^{-1} \hat{\rho}\le L_\alpha(\#t,\#s)$ and $U_\alpha(\#t,\#s)\le c\hat{\rho}$ are met, leading us to:
\begin{align}
\label{eq:sampleuntil}
    \Prob{c^{-1} \rho\le \hat{\rho}\le c \rho}
    &= \Prob{c^{-1} \hat{\rho}\le \rho\le c \hat{\rho}}\\
    &\ge \Prob{L_\alpha(\#t,\#s)\le \rho\le U_\alpha(\#t,\#s)}
    \ge 1-\alpha. \notag
\end{align}

In our experiments where we had 1,000 – 1,000,000 trials, we found out that, irrespective of the number of trials, 88 successes were sufficient to satisfy the condition stated with the Clopper-Pearson interval (Equation 4) for $\alpha = 0.05$ and $c = 1.25$.

For hard cases in which achieving the condition may require an unreasonable amount of computation, we terminate the tree sampling immediately and use the graph sampling described in \Cref{subsec:CGSampling}. For the experiments, we stop if there are no more than 10 successes for 50,000 trials.
Since 88 successes are required to meet the stated condition and we only continue the tree sampling if there are at least 11 successes in the first 50,000 trials, we expect that the number of trials won’t exceed 400,000 in the tree sampling.

\vspace{-0.5mm}
\begin{restatable}{theorem}{ThmTreeUnbiased}
\label{thm:tree_unbiased}
    \Cref{alg:treesampling} is an unbiased, consistent estimator for subgraph cardinality.
\end{restatable}

\section{Stratified Graph Sampling}
 
\label{subsec:CGSampling}
While the presented tree sampling algorithm gives accurate and efficient estimate for most cases, there are some hard instances where sampling space remains large compared to the number of embeddings even after filtering. 
This is due to the fact that the number of candidate trees can be asymptotically larger than $AGM(q)$, which is a tight upper bound for the number of embeddings~\cite{SSTE}.

To address such an issue, we develop a \textit{stratified graph sampling} algorithm with the following components. 
(1) We consider non-tree edges together during the sampling phase, ensuring that the sample space size is bounded by $O(AGM(q))$.
(2) We propose stratified sampling to select samples from diverse regions within the sample space. 
By combining these two components, our algorithm achieves high accuracy and efficiency, particularly on difficult instances.

\subsection{Extendable Candidates}
We define a vertex mapping $M$ as a \textit{partial embedding} when it is an embedding of an induced subgraph $q'$ of $q$ to $G$. 
Given a partial embedding $M$, a vertex $u\in V_q\setminus V_{q'}$ is an \textit{extendable vertex} if there exists $u'\in N_q(u)$ such that $u'\in V_{q'}$.
We define the set of \textit{extendable candidates} of $u$ regarding a partial embedding $M$ as follows:
\begin{align*}
    C_M(u) = \bigcap_{\set{(u', v') \in M\mid u' \in N_q(u)}} &C(u\mid u',v')-\text{image of }M.
\end{align*}
where image of $M$ is the vertices of the data graph that are mapped by $M$, i.e., $\set{v \in V_G \mid  {^\exists u} \in V_{q'} \ \text{such that}\  M(u) = v}$.
When $M = \emptyset$, $C_M(u)$ is set to be $C(u)$.
$C_M(u)$ is computed by the multi-way set intersection over candidate neighbor sets of each neighbor $u'$ of $u$. 
To maintain injectivity, candidate vertices that are already mapped by $M$ are removed from $C_M(u)$.
For $v\in C_M(u)$, $M\cup \set{(u,v)}$ forms a partial embedding since $M\cup \set{(u,v)}$ is an embedding of an induced subgraph of $q$ with vertex set $V_{q'}\cup \set{u}$ to $G$.

For example, in \Cref{subfig:CTSb}, at the time when $M = \set{(u_1, v_2), (u_2, v_4)}$, $u_4$ is an extendable vertex with extendable candidates $C_M(u_4) = \set{v_8, v_9}$ as they are the vertices in $C(u_4)$ which are candidate neighbors for both $(u_1, v_2)$ and $(u_2, v_4)$.

\subsection{Stratified Graph Sampling}
\label{subsec:graphsampling}
We say that a partial embedding \textit{extends} $M$ if it contains $M$ as a subset. Let $w_M$ denote the number of embeddings of $q$ to $G$ extending the partial embedding $M$. 
It is clear that $w_M = 1$ when $\abs{M} = \abs{V_q}$, and $w_M = 0$ when $C_M(u) = \emptyset$. For other cases, we employ the method of \textit{stratified sampling} to obtain an unbiased consistent estimator for $w_M$.

\begin{algorithm}[t]
    \caption{Stratified Graph Sampling}
    \label{alg:CGSampling}
    \rm\sffamily
    \SetKwFunction{FSample}{EstimateW}
    \SetKwFunction{FCGSample}{CandidateGraphSampling}
    \SetKwProg{Fn}{Function}{:}{}

    \KwIn{The compact CS, a query graph $q$, the number of successful tree samples \#successes}
    \KwOut{Estimated number of embeddings}
    \Fn{\FCGSample{\rm \sffamily CS, $q$, \#successes}}{
        $ub_\emptyset$ $\gets$ GetSampleSize($q$, \#successes)\; 

        (estimate, \#trials) $\gets$ EstimateW(CS, q, $\emptyset$, $ub_\emptyset$)\;
        
        \Return{\rm \sffamily estimate\;}
    }
    \vspace{0.6em} 
    \KwIn{The compact CS, a query graph $q$, a current partial embedding $M$, an upper bound for sample size $ub_M$}
    \KwOut{Estimated number of embeddings extending $M$, number of samples used}
    \Fn{\FSample{\rm \sffamily CS, $q$, $M$, $ub_M$}}{

        \lIf{$\abs{M} = \abs{V_q}$}{
            \Return{$(1, 1)$}
        }
        $u$ $\gets$ ChooseExtendableVertex($q$, $M$)\;
    
        $C_M(u)$ $\gets$ GetExtendableCandidates(CS, $q$, $u$)\;
    
        \lIf{\rm\sffamily $C_M(u)$ is empty}{
           \Return{$(0, 1)$}
        }
        sz $\gets$ $\min(C_M(u) / k, ub_M)$\; 
        \label{line:GetGroups}
        
        $S$ $\gets$ GetRandomSubset($C_M(u)$, sz)\;

        ($\hat{w}_M$, \#trials) $\gets$ $(0, 0)$\;
        \For {$i = 1 \dots \abs{S}$} {
            $M' \gets M \cup \set{(u, S_i)}$\; 
            $(\hat{w}_{M'}, n)$ $\gets$ EstimateW(CS, $q$, $M'$, $\frac{ub_M - \mathsf{\#trials}}{\abs{S} - i + 1}$)\;
            \label{line:upperbound}
            ($\hat{w}_M$, \#trials) $\gets$ ($\hat{w}_M + \hat{w}_{M'}$, \#trials + $n$)\;
        }
        $ \hat{w}_M \gets \frac{\abs{C_M(u)}}{\abs{S}}~\times \hat{w}_M$\;
        \Return{\rm \sffamily ($\hat{w}_M$, \#trials)\;}}
\end{algorithm}

\paragraph{Stratified Sampling}
\textit{Stratification} is a strategy of dividing the population into multiple subpopulations that are mutually exclusive and jointly exhaustive, namely \textit{strata}, prior to sampling in order to ensure diversity of resultant samples~\cite{Sampling}. 
In our instance, to estimate $w_M$, the population is the set of partial embeddings extending $M$. For an extendable vertex $u$ and candidates $C_M(u)$ for $M$, the population can be then stratified into subpopulations (called \textit{groups}) of partial embeddings extending $M \cup \set{(u, v)}$ for each $v \in C_M(u)$. 
Let $\hat{w}_M$ be an estimate for $w_M$. By definition,
\begin{equation}
    w_M = \sum_{v \in C_M(u)} w_{M \cup \set{(u, v)}}.
\end{equation}
Therefore, by taking a random subset $S$ of candidates from $C_M(u)$, and estimating the values of $w_{M \cup \set{(u, v)}}$ for $v \in S$, we obtain an unbiased estimator for $w_M$.
Formally, consider a random subset $S$ of $C_M(u)$, where the probability of each $v \in C_M(u)$ being included in $S$ is uniform as $\abs{S} / \abs{C_M(u)}$. We define $\hat{w}_M$ as 
\begin{equation}
\label{eq:RecursiveEstimate}
    \hat{w}_M = \frac{\abs{C_M(u)}}{\abs{S}}\sum_{v \in S} \hat{w}_{M \cup \set{(u, v)}}.
\end{equation}
If $\abs{M} = \abs{V_q}$, $\hat{w}_M$ is set to $1$, and if $C_M(u) = \emptyset$, $\hat{w}_M$ is set to $0$. Each $\hat{w}_{M \cup (u, v)}$ are computed recursively.

\noindent The pseudocode of the algorithm is presented in \Cref{alg:CGSampling}.

\paragraph{Sampling Order}
In each step, we choose an extendable vertex $u \in V_q$, and compute extendable candidates $C_M(u)$ given partial embedding $M$.  
To choose the extendable vertex with the number of extendable candidates as small as possible, we choose the vertex $u$ which is not in $M$ and has the highest number of neighbors in $M$ in each step. 
At first, we choose the vertex with smallest $\abs{C(u)}$. 
We note that for given $q$ and the compact CS, the order of extendable vertices is deterministic.

\paragraph{Adaptive Allocation}
Examining a fixed proportion of extendable candidates in each recursive step by taking $\abs{S} = \abs{C_M(u)} / k$ for some constant $k$ may result in an exponentially large number of samples.
To circumvent this issue, we limit the maximum number of samples used for estimating $w_M$ as $ub_M$ for each $M$. 
We choose the subset $S$ to be a random subset with a size of $\min(\abs{C_M(u)} / k, ub_M)$ (\Cref{alg:CGSampling}, \Cref{line:GetGroups}). 
The value of $ub_{\emptyset}$ is determined based on the difficulty of the query, with further details to be discussed later.
We then set the upper bound $ub_{M \cup \set{(u, v)}}$ for each $v \in S$ in such a way that the sum of upper bounds does not exceed $ub_M$. 

For the first $v \in S$, we assign $ub_{M \cup \set{(u, v)}}= \frac{ub_M}{\abs{S}}$. The recursive call on $M \cup \set{(u, v)}$ returns an estimate for $w_{M \cup \set{(u, v)}}$ and the actual number of samples taken throughout the recursion.
The recursion on $M$ may terminate before the number of samples reaches $ub_M$ (e.g., if $\abs{M} = \abs{V_q}$ or $C_M(u) = \emptyset$, only one sample is encountered regardless of how large $ub_M$ is).
In such cases, we adaptively increase the sample size for subsequent calls equally using this information to enhance the overall accuracy.
The upper bound for the $i$-th $v \in S$ is determined as the remaining upper bound (i.e., $ub_M$ minus the number of samples seen from the first to the $(i-1)$-th) divided by the number of remaining calls ($\abs{S} - i + 1$), as in \Cref{line:upperbound}.
As this can only increase the upper bound from $ub_M / \abs{S}$, we ensure that each $M \cup \set{(u, v)}$ receives at least $ub_M / \abs{S}$ for the upper bound, thereby obtaining diverse samples even with limited sample sizes.

\begin{figure}
    \centering
    \includegraphics[width=\linewidth]{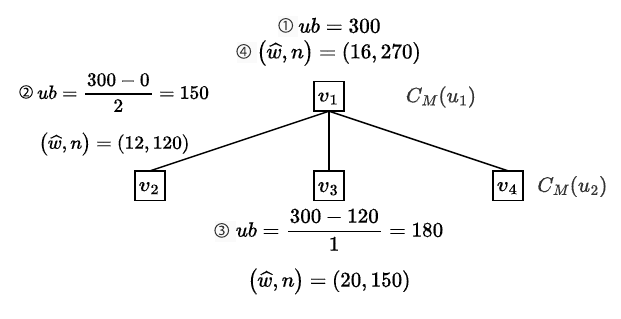}
    \caption{Example of Stratified Sampling}
    \label{fig:stratify}
\end{figure}

\begin{example}
    An example with $ub_M = 300$ is depicted in \Cref{fig:stratify}. 
    For $C_M(u_2) = \set{v_2, v_3, v_4}$, suppose $S = \set{v_2, v_3}$ is the randomly chosen subset.
    For the group with $v_2$, the upper bound $ub_{M \cup \set{(u_2, v_2)}}$ is set to 150 as $ub_M$ is expected to split into two groups. 
    Suppose that for the first group, the returned estimate is 12 with 120 samples used.
    The upper bound for the second group is then determined as 180, since only 120 samples were used. Assume that the group returned 20 with 150 samples used. Aggregating the results, the initial call on $v_1$ returns 16 (average of estimates returned by two subcalls) and reports that 270 samples are actually used.
\end{example}

\begin{restatable}{theorem}{ThmGraphUnbiased}
\label{thm:graph_unbiased}
\Cref{alg:CGSampling} is an unbiased consistent estimator for subgraph cardinality, and the worst-case time complexity is bounded by $O(AGM(q))$.
\end{restatable}

\paragraph{Sample Size} 
It is favorable to use more samples on difficult instances for accuracy, and use a small sample size on easier instances for efficient estimation.
However, measuring the difficulty is itself challenging as it depends on the ratio between the number of embeddings and the size of the sample space.

The query size is commonly used in prior works on graph homomorphism counting, as a larger query requires more samples in general \cite{GCare, Alley}.
Empirically, while the difficulty of an instance increases with the query size, the size of the data graph does not seem to be strongly relevant.
Since we perform the graph sampling only when the tree sampling has not found enough successes after some number of samples, the number of successes from the tree sampling can act as a proxy measure of the difficulty of an instance. 
Hence we take larger $ub_\emptyset$ when the number of successes from the tree sampling is small. 
However, this association is not linear empirically, and therefore we take the square root on the number of successes, setting $ub_\emptyset$ as 

\begin{equation}
    \label{eq:graph-sample-size}
    ub_{\emptyset} = \frac{\abs{V_q} \times K}{\sqrt{\text{\#\ successes} + 1}}
\end{equation}
for some constant $K$. Using $K = 100,000$, we achieve a reasonable accuracy over datasets of wildly different statistics.

\section{Performance Evaluation}
\label{sec:experiments}

\subsection{Experimental Setup}
We conduct experiments to evaluate the performance of our algorithm, \Fastest, by answering the following questions.
\begin{itemize}
    \item \textbf{Accuracy}. Compared to state-of-the-art competitors, how accurate is \Fastest on various real-world datasets and difficult query instances with diverse characteristics? (\Cref{subsec:accuracy})
    \item \textbf{Efficiency}. Compared to state-of-the-art competitors, how efficient is \Fastest in terms of execution time and memory usage? (\Cref{subsec:efficiency})
    \item \textbf{Evaluation of Techniques}. How do our safety conditions and tree/graph sampling algorithms contribute to accurate and efficient estimation? (\Cref{subsec:ablationstudy})
\end{itemize}
We mainly compare our algorithm with \Alley \cite{Alley}, since it has significantly outperformed the synopses-based algorithms and sampling-based algorithms such as \textsf{Characteristic Sets} \cite{CSET}, \textsf{SumRDF} \cite{SumRDF}, \textsf{Correlated Sampling} \cite{CorrelatedSampling}, \textsf{JSub} \cite{JSub}, and \textsf{WanderJoin} \cite{WanderJoin}. Since \Alley was originally developed for estimating the number of homomorphic embeddings, we modified it for isomorphic embeddings. 
We list the variants of \Alley we compared to as follows.
\begin{itemize}
\item \Alley: We modified \Alley so that for each random walk sampled, whether the walk is injective is additionally checked.
\item \AlleyPlus: We modified \Alley to not choose an already visited data vertex twice during each random walk, ensuring that only injective walks are sampled.
\item \AlleyTPI: We modified the Tangled Pattern Index (TPI) structure proposed in \cite{Alley} for isomorphism and built it on top of \AlleyPlus. 
Following \cite{Alley}, we set the maximum pattern size to 4 for vertex labeled datasets, and 5 if both vertex labels and edge labels are present.
\end{itemize}
To compare our algorithm with GNN-based methods, we consider two recent state-of-the-art works, \LSS \cite{LSS} and \NeurSC \cite{NeurSC} as two additional baselines. 

\paragraph{Environment} \Fastest is implemented in C++. The source code of \Alley and \LSS were publicly available, and \NeurSC was obtained from the authors. Hyperparameters were set to their suggested default values. 
Experiments were conducted on a machine with two Intel Xeon Silver 4114 2.20GHz CPUs, an NVIDIA RTX 3090 Ti GPU, and 256GB memory. 

\paragraph{Datasets} We conduct experiments on real-world large-scale graph datasets used in previous works \cite{RapidExperiment, DAF, VEQ, NeurSC, Alley}. 
For fair comparison, query graphs were obtained from a comparative study on subgraph matching algorithms \cite{RapidExperiment} if possible. Large query graphs for the \Dblp dataset ($\abs{V_q}=12, 20$) and query graphs for the \Yago dataset were generated via random walk. Since \Yago has edge labels, we extended our algorithm to handle edge-labeled graphs. We used an exact subgraph matching algorithm \DAF \cite{DAF} for obtaining the ground truth results. Since the computational cost for exact counting is extremely high, we only used the query graphs for which the exact count could be computed within 2 hours. 
The data graphs are listed in \autoref{tab:datastat}.

\begin{table}[t]
    \caption{Statistics of Data Graphs. $\Sigma_V$ and $\Sigma_E$ denote the sets of vertex labels and edge labels, respectively.}
    \centering
    \resizebox{\linewidth}{!}{
    \begin{tabular}{l r r r r r r}
    \toprule
    Dataset & $\abs{V}$ & $\abs{E}$ & $\abs{\Sigma_V}$& $\abs{\Sigma_E}$  & \#Query & $\abs{V_q}$\\\midrule
    Yeast (Ye)   & 3.1K       & 12.5K      & 71 & -            & 1,707   & 4 to 32       \\
    WordNet (Wo)   & 76.9K      & 120.4K     & 5  & -          & 1,164   & 4 to 20       \\
    DBLP (Db) & 317.1K & 1.0M & 15 & - & 911 & 4 to 20\\
    Youtube (Yo)    & 1.1M     & 3.0M   & 25   & -          & 922     & 4 to 32      \\
    Patents (Pa) & 3.8M & 16.5M & 20 & - & 1,485 & 4 to 32 \\ 
    YAGO (Ya) & 12.8M & 15.8M & 185K & 91 & 960 & 4 to 20 \\ 
    \bottomrule
    \end{tabular}
    }
    \label{tab:datastat}
\end{table}

\paragraph{Performance Measure} 
Accuracy is measured by \textit{q-error}~\cite{qerror}, defined as $\max\left(\frac{\max(1, w)}{\max(1, \hat{w})}, \frac{\max(1, \hat{w})}{\max(1, w)}\right)$, where $\hat{w}$ is the estimate and $w$ is the ground truth number of embeddings, as done in previous works \cite{GCare, Alley, LSS, NeurSC}.
Efficiency is measured by average elapsed time per query. For \LSS and \NeurSC, we evaluated with five-fold cross-validation.

For the tree sampling, we set the acceptable failure probability to $\alpha = 0.05$ and the tolerable relative error to $c = 1.25$. Using these parameters, the sample size for the tree sampling is adaptively determined by the Clopper-Pearson interval (\autoref{eq:sampleuntil}).

\begin{figure*}[t]
    \centering
    \begin{subfigure}{\textwidth}
        \centering
        \includegraphics[width=1\textwidth]{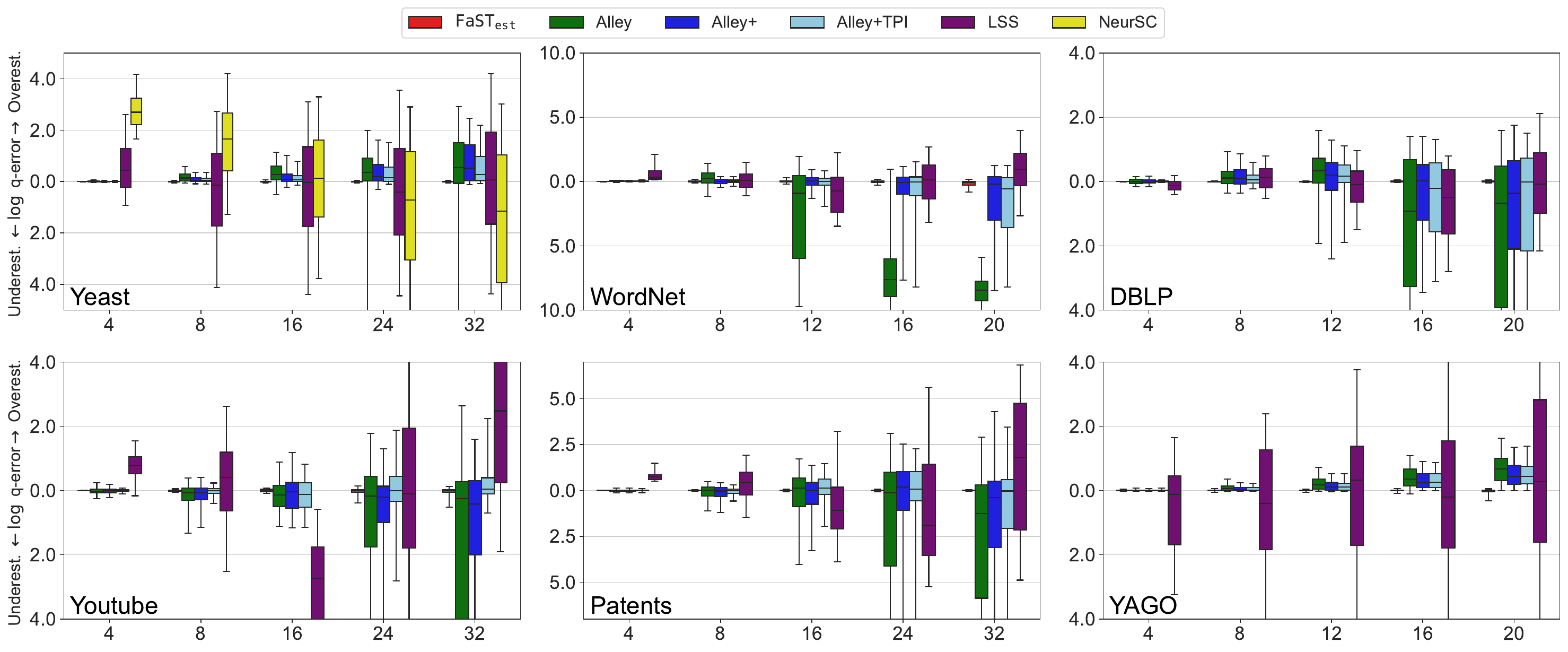}
        \subcaption{log\textsubscript{10} q-error with the x-axis representing the number of query vertices.}
        \label{fig:acc-qsize}
    \end{subfigure}
    \begin{subfigure}{\textwidth}
        \centering
        \includegraphics[width=1\textwidth]{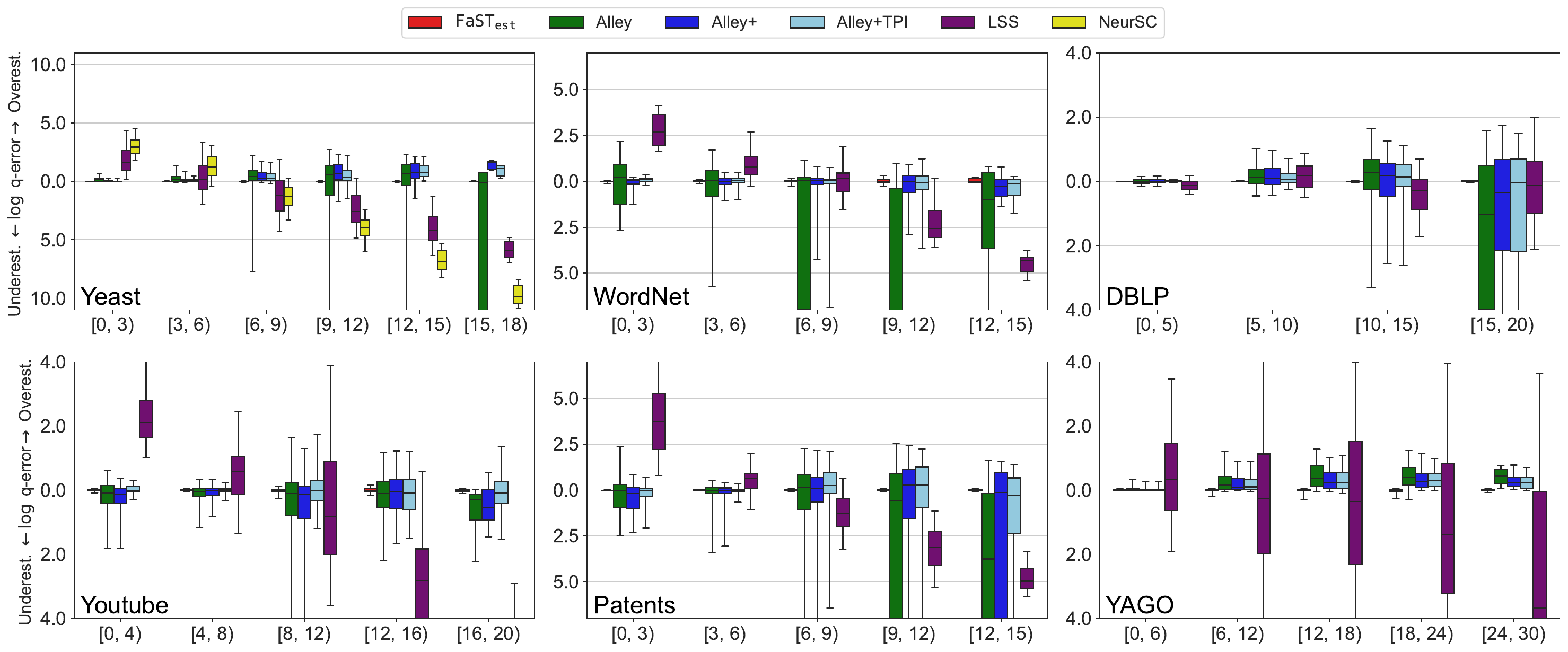}
        \subcaption{log\textsubscript{10} q-error with the x-axis representing log\textsubscript{10} (\#embeddings) of the query graph.}
        \label{fig:acc-truth}
    \end{subfigure}
    \caption{q-error evaluation on real world datasets (closer to 0 on the y-axis indicates better performance). The boxplot presents values of log q-error of a set of queries. The lower and upper whisker denotes the 5\% and 95\% quantile, respectively. The box itself represents the interquartile range, covering values between the 25\% and 75\% quantiles. The median of each set of queries is represented by the black line within the box.}
    \label{fig:accuracy}
\end{figure*}

\subsection{Accuracy}\label{subsec:accuracy}
\Cref{fig:accuracy} presents the accuracy of \Fastest compared to state-of-the-art competitors on real-world graph datasets, varying the size of the query graph (\Cref{fig:acc-qsize}) and the number of embeddings (\Cref{fig:acc-truth}).
In \Cref{fig:accuracy}, \textit{difficult} instances for subgraph cardinality estimation are large queries or queries with many embeddings.
To differentiate overestimation and underestimation, if $\hat{w} \geq w$, $\log_{10} \frac{\max(1, \hat{w})}{\max(1, w)}$ is shown above $y = \log_{10} 1 = 0$; if $\hat{w} < w$, $\log_{10} \frac{\max(1, w)}{\max(1, \hat{w})}$ is shown below $y = 0$, as in \cite{NeurSC, LSS, Alley, GCare}.
To ensure a fair comparison, we present the results with $K = 100,000$ (\autoref{eq:graph-sample-size}) for our graph sampling, and choose the sample size of \Alley so that its overall execution time is not less than ours. 
Except for \Yeast, we were unable to train \NeurSC due to time and memory limits, and therefore we omit the results.
Across data graphs with diverse statistical properties and query sizes ranging from small (4) to large (32), \Fastest significantly outperforms both sampling methods and GNN methods in terms of accuracy. Note that a log q-error of 1 represents tenfold over/underestimation.

\Cref{fig:acc-qsize} shows that the q-error of GNNs increases as the size of the query graph grows. 
For example, on \Patents-32 (queries with size 32 on the \Patents dataset), \Fastest shows an average q-error of 1.04 while \LSS shows $4328.3$ ($\log$ q-error of 0.02 and 3.64, respectively), representing an improvement of more than three orders of magnitude. Compared to \NeurSC, \Fastest shows an average q-error of 1.07 on \Yeast-32, where \NeurSC shows $811.5$ ($\log$ q-error of 0.03 and 2.91, respectively).

Additionally, GNNs tend to overestimate when the true cardinality is small, while they underestimate when the true cardinality is large, as observed in \Cref{fig:acc-truth}.
While GNN-based methods can make fast inference for each query, the accuracy on difficult instances is substantially lower than that of sampling-based algorithms. 
Therefore, for further evaluation, we focus on comparison with the state-of-the-art sampling algorithm, \Alley. 

Sampling algorithms are often prone to sampling failure, where no successful sample is found and the algorithm consequently reports zero as the estimate. \Alley also exhibits sampling failure on difficult instances, such as large queries on \Patents, \Youtube, and \WordNet datasets. However, \Fastest significantly reduces the sample space by strong filtering, thereby achieving high accuracy on such instances.
Especially, on \WordNet-20, \Fastest shows an average q-error of 1.81 (log q-error of 0.26), where \AlleyTPI exhibits an average q-error of 418.5 (log q-error of 2.62), which represents an improvement of more than two orders of magnitude.

The accuracy of \Alley tends to degrade as the true cardinality increases. Yet, as seen in \Cref{fig:acc-truth}, \Fastest demonstrates impressive accuracy even in cases where the true cardinality is exceptionally large. For the \Yeast dataset, on queries with more than $10^{15}$ embeddings (represented by $[15, 18)$), \Fastest exhibits an average q-error of $1.10$, which corresponds to a relative error of about 10\%.

\begin{figure*}[t]
    \centering
    \includegraphics[width=\linewidth]{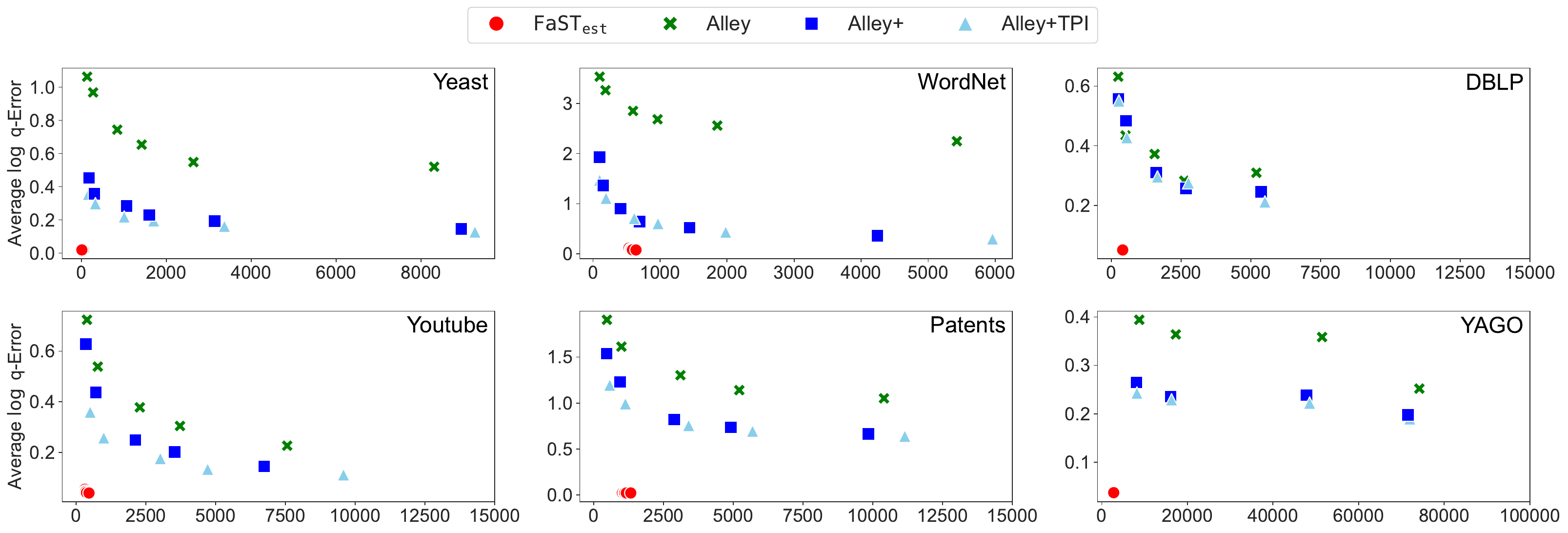}
    \captionsetup{justification=centering}
    \caption{Accuracy-Time Tradeoff with different sample sizes. (Lower = Better accuracy, Left = Faster query processing)\\The x-axis represents the average elapsed time per query (ms).}
    \label{fig:tradeoff}
\end{figure*}

\subsection{Efficiency}
\label{subsec:efficiency}
\paragraph{Query Processing Time} 
\Cref{tab:time-breakdown} presents both indexing time and query processing time of \Fastest. Specifically, the query processing time is broken down into filtering time, tree counting time, tree sampling time, and graph sampling time.
The values for Filtering, Tree Counting, Tree Sampling, and Graph Sampling (A) represent the average time spent per query (i.e., total graph sampling time divided by the number of queries).
Graph Sampling (H) shows the average graph sampling time per hard query (i.e., total graph sampling time divided by the number of hard queries), where a \textit{hard} query means a query that proceeds to the graph sampling because no more than 10 successes were observed in 50,000 trials of the tree sampling. For example, in \Yeast, 25 out of 1,707 queries are hard queries.
While the graph sampling usually takes much longer compared to the tree sampling, its overall impact on the query processing time remains limited because the majority of queries are resolved through the tree sampling.
In general, filtering takes the most time, especially for large datasets.

\setlength{\tabcolsep}{3pt}

\begin{table}[b]
    \caption{Breakdown of the running time of \Fastest (ms). Hard queries denote the queries that required graph sampling.}
    \resizebox{\linewidth}{!}{
    \centering
    \begin{tabular}{l r r r r r r}
    \toprule
    Dataset & Yeast & WordNet & DBLP & Youtube & Patents & YAGO\\\midrule
    Indexing  & 658 & 509 & 49093 & 3419 & 16066 & 13512 \\
    Query Processing & 28.06 & 555.53 & 403.94 & 382.08 & 1122.13 & 2870.08  \\\midrule
    Filtering & 4.70 & 443.96 & 400.67 & 293.83 & 980.89 & 2844.65 \\
    Tree Counting & 0.21 & 23.75 & 2.04 & 1.45 & 2.72 & 7.20 \\
    Tree Sampling & 4.47 & 42.48 & 1.23 & 21.27 & 11.97 & 6.09 \\
    Graph Sampling (A) & 18.68 & 45.54 & 0 & 65.53 & 89.55 & 12.14 \\
    Graph Sampling (H)  & 1275.47 & 84.01 & 0 & 663.94 & 929.94 & 3884.80 \\
    \# Hard Queries & 25/1707 & 631/1164 & 0/911 & 91/922 & 143/1485 & 3/960\\
    \bottomrule
    \end{tabular}
    }
    \label{tab:time-breakdown}
\end{table}

\setlength{\tabcolsep}{6pt}

\paragraph{Accuracy-Time Tradeoff}
Sampling algorithms have a natural tradeoff between accuracy and query processing time, as the estimation can be more accurate by taking a larger sample size. In \Cref{fig:tradeoff}, we demonstrate the performance of our algorithm in this tradeoff by conducting an experiment varying sample sizes.

For the graph sampling, we vary the sample size by changing value of $K$ in \autoref{eq:graph-sample-size} from $5 \times 10^3$ to $3 \times 10^5$. 
We determine the sample sizes for \Alley as $K \times \abs{V_q}$, using the same range of $K$. 
Each mark in \Cref{fig:tradeoff} represents the elapsed time and the log q-error averaged over all queries at different sample sizes.
As the sample size increases, the sampling algorithms tend to require more time (moving from left to right on the x-axis), while the average log q-error decreases, indicating improved accuracy (moving from top to bottom on the y-axis).
For large datasets, the results of \Alley variants taking an excessive amount of time were omitted (more than 15,000 ms for \Dblp, \Youtube, \Patents, and 100,000 ms for \Yago per query on average).

\Cref{fig:tradeoff} highlights the accuracy-to-time tradeoff of \Fastest when compared to \Alley.
As seen in \Cref{tab:time-breakdown}, most queries on \Yeast, \Dblp, and \Yago are handled by the tree sampling, whose sample size is determined by \autoref{eq:sampleuntil}, independent of $K$. This explains the result appearing almost as a single dot for \Yeast, \Dblp, and \Yago. In these datasets, the tree sampling demonstrates outstanding accuracy and efficiency.
For \WordNet, \Youtube, and \Patents, we observe that increasing the sample size improves the accuracy for hard queries, resulting in a reduction of the average log q-error.

Although \Fastest incurs some overhead due to filtering for large datasets, it remains significantly more efficient in terms of the required execution time to achieve a specific accuracy level.
In \WordNet, \AlleyTPI, the top competitor, takes about 6,000 milliseconds per query to attain an average log q-error of 0.3 (q-error of about 2). On the other hand, \Fastest reaches an average log q-error of 0.1 (q-error of about 1.26) with just about 600 milliseconds per query.

The efficiency of \Fastest is achieved by the synergy of our three key components. 
Firstly, our strong filtering substantially reduces the sample space before sampling, which reduces the number of required samples. Filtering also reduces the cost of obtaining a sample.
Secondly, the tree sampling generates accurate results without requiring costly intersection, addressing most cases with high efficiency.
Lastly, the graph sampling tackles the remaining hard cases.

\begin{table}[b]
    \centering
    \caption{Indexing time of Triangles, 4-Cycles and TPI (sec)}
    \resizebox{\linewidth}{!}{
    \begin{tabular}{ccccccc}
    \toprule
    &           Ye      & Wo  & Db & Yo    & Pa & Ya      \\ \toprule
    Triangles   & 0.006   & 0.04 & 1.3 & 3.4 & 16.1 & 13.5 \\
    Four-Cycles & 0.65   & 0.47  & 47.8 & - & - & -  \\ 
    TPI & 278.0    & 41.5 & 5467.1 & 6069.8 &  24837.6 & 4732.7 \\ 
    \bottomrule
    \end{tabular}
    }
    \label{tab:enumeration-overhead}
\end{table}

\paragraph{Indexing Time}
\autoref{tab:enumeration-overhead} shows the indexing time for triangles and 4-cycles for our datasets compared to the TPI (Tangled Pattern Index) from Alley \cite{Alley}.
In our experiments, all six datasets utilize Triangle Safety, and three datasets (\Yeast, \WordNet, and \Dblp) use Four-Cycle Safety.
Compared to TPI which takes exponential time, our indexing can be done much more efficiently in polynomial time. 
We note that since this indexing is required only once per data graph, it can be done in advance to process many queries.

\begin{figure}[t]
    \centering
    \includegraphics[width=\linewidth]{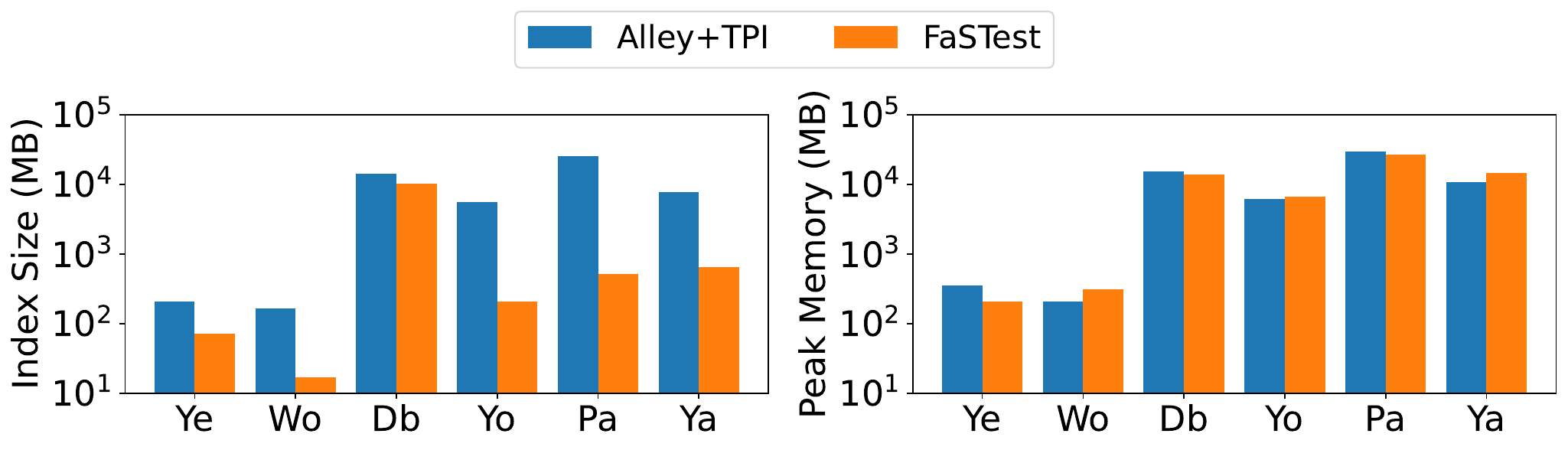}
    \caption{Index size and peak memory consumption.}
    \label{fig:memory-usage}
\end{figure}

\paragraph{Memory Usage}
\Cref{fig:memory-usage} shows the size of the index and peak memory consumption (measured as the highest memory usage observed across all queries, including both the index and the Candidate Space) of \Fastest compared to \Alley. 
While the Candidate Space is built for each query, the index is built only once per data graph.
\Fastest builds a smaller index compared to \Alley's TPI structure, and its peak memory usage remains comparable to that of \Alley.

\subsection{Evaluation of Techniques}
\label{subsec:ablationstudy}

\begin{figure}[t]
    \centering
    \includegraphics[width=\linewidth]{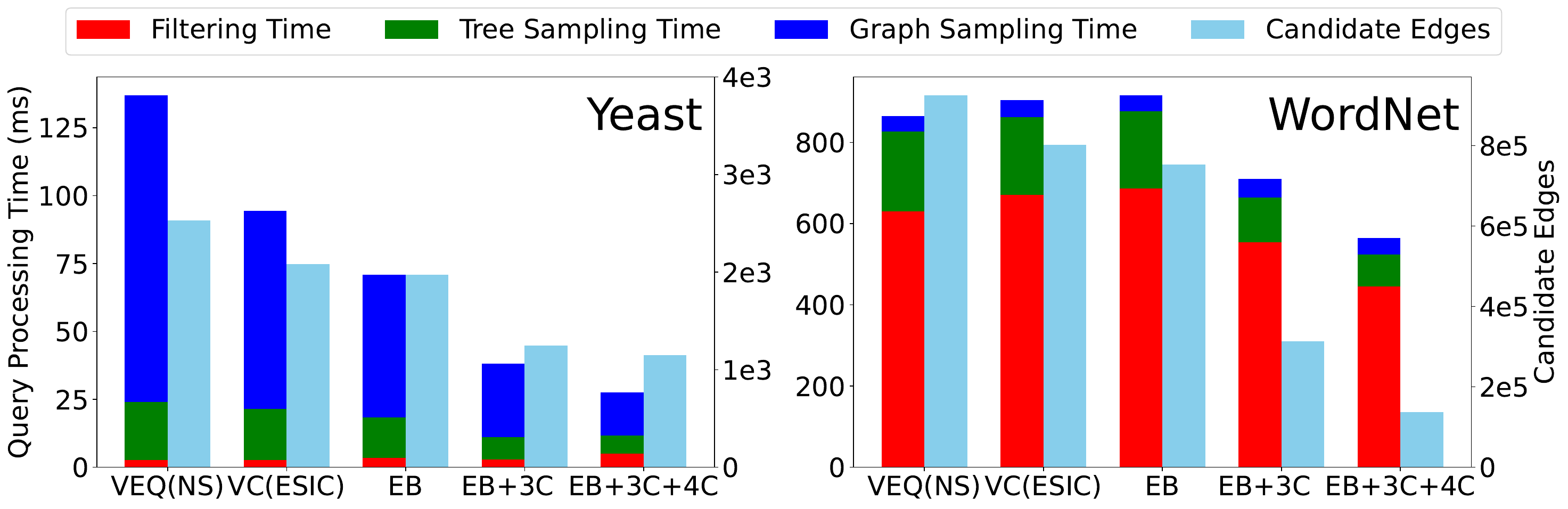}
    \caption{Query processing time and size of the Candidate Space after filtering, varying safety conditions.}
    \label{fig:ablation-revised}
\end{figure}

\paragraph{Safety Conditions}
    \Cref{fig:ablation-revised} shows the breakdown of the query processing time and the number of candidate edges after filtering for \Yeast and \WordNet, varying only the safety conditions used. The red, green, and blue portion of the left bar indicates the time spent on filtering, tree sampling, and graph sampling, respectively. The light blue bars show the number of candidate edges after filtering.

    Together with Edge Bipartite Safety (EB), Triangle Safety (3C) and Four-Cycle Safety (4C), we implemented the safety conditions from two state-of-the-art subgraph matching algorithms: Neighbor Safety (NS) proposed by \VEQ~\cite{VEQ} and Exact Star Isomorphism Constraint (ESIC) by VC~\cite{VCSubgraphMatching} as the baselines.

    On \Yeast, a stronger safety condition significantly reduces the number of candidate edges, at the expense of slight increase in the filtering time. 
    This reduction in candidate edges accelerates both tree sampling and graph sampling, thereby improving overall query processing time. 
    For \WordNet, using EB increases the filtering time compared to ESIC. 
    However, combining EB and substructure conditions (3C, 4C) reduces the filtering time as well, since the candidate edges removed by strong safety conditions are not considered in subsequent filtering steps, reducing the overall cost of filtering.

Specifically, our filtering method combining all three safety conditions reduces the number of candidate edges in \WordNet by 80\% compared to NS and by 75\% compared to ESIC.
On both datasets, each individual safety condition we propose has a clear effect in reducing query processing time and candidate edges.

\begin{table}[b]
\caption{Geometric average of number of candidate trees (Lower is better)}
\resizebox{\linewidth}{!}{
\begin{tabular}{lcccccccc}
\toprule
            & Ye             & Wo             & Db              & Yo             & Pa             & Ya              \\ 
\toprule
\#Edge MST  & 2.0e7 & 6.4e13 & 3.2e8 & 5.6e10 & 5.3e7 & 2.6e12\\
Random Walk & 7.8e6 & 4.6e12 & 3.0e8 & 8.2e9 & 4.0e7 & 1.6e11\\
Density MST & \textbf{2.5e6} & \textbf{1.0e12} & \textbf{2.7e8} & \textbf{1.4e9} & \textbf{2.6e7} & \textbf{2.9e10}  \\
\bottomrule
\end{tabular}
}
\label{tab:numtrees}
\end{table}

\paragraph{Effectiveness of Density Heuristic}
To reduce the number of candidate trees, one can consider minimizing the product of number of remaining candidate edges (\#Edge MST), while the density heuristic tries to minimize the product of \textit{density}. Additionally, we compare our strategy to the random walk-based spanning tree generation (Random Walk). It is proven that with the random walk, one can generate a random spanning tree with uniform probability~\cite{RandomWalkSpanningTree}.

\autoref{tab:numtrees} compares the geometric average of the number of candidate trees with each strategy. 
We observe that \#Edge MST generates a significantly higher number of candidate trees, while the density heuristic (Density MST) consistently produces better results.
The differences between strategies are larger when the number of candidate trees is large.
Density MST significantly reduces the sample space by up to 81.8\% compared to Random Walk and by up to 98.8\% compared to \#Edge MST (in \Yago), which directly translates into a similar reduction in the number of samples required.

\begin{figure}[t]
    \centering
    \includegraphics[width=\linewidth]{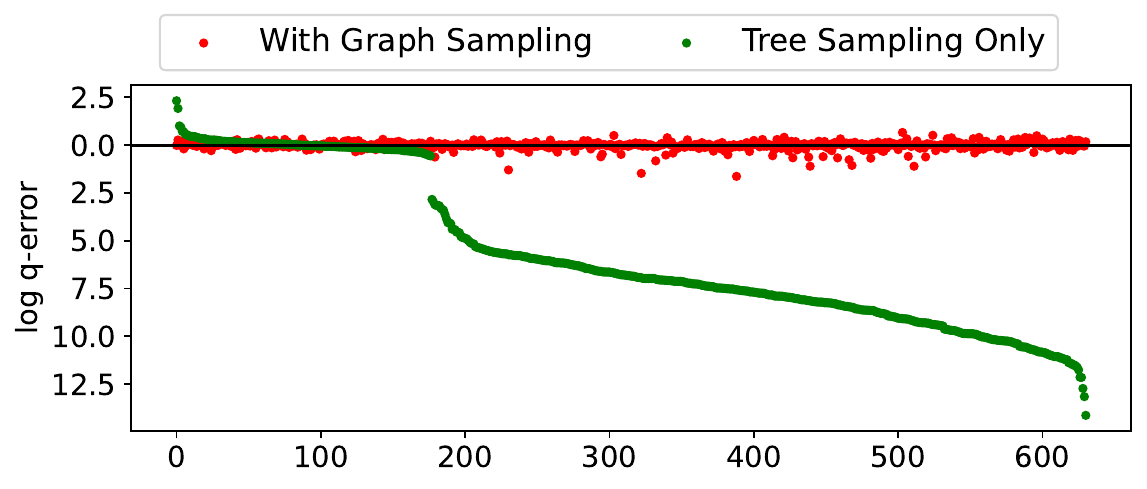}
    \caption{Comparative analysis of log q-error on \WordNet with (red) and without (green) the graph sampling. The x-axis indicates queries sorted by q-error obtained from the tree sampling. Values closer to 0 is better, with dots above (resp. below) black line representing overestimation (resp. underestimation).}
    \label{fig:only-tree}
\end{figure}

\paragraph{Comparison of Tree and Graph Sampling}
As shown in \Cref{tab:time-breakdown}, the average tree sampling time per query is significantly lower than the average graph sampling time per hard query. However, graph sampling is necessary for hard queries as tree sampling alone may not yield accurate estimates within a reasonable number of samples. 
In the extreme case, there is a query which has $1.5 \times 10^6$ embeddings in \WordNet and for which $1.2 \times 10^{26}$ candidate trees exist, necessitating about $7\times 10^{21}$ samples to get 88 successes, which is required for reaching the desired accuracy.
In the experiment, \Fastest proceeded to the graph sampling after 50,000 trials of the tree sampling as it failed to find any successes.

\autoref{fig:only-tree} shows the log q-errors with and without the graph sampling for 631 hard queries on \WordNet that required the graph sampling. The x-axis indicates queries from those where the tree sampling overestimated to those where it underestimated. Each query is represented by a red dot and a green dot, denoting the log q-errors with and without the graph sampling, respectively.
The results indicate that the tree sampling can lead to severe overestimation when one or two extremely unlikely successes occur, or severe underestimation when there are no successes. 
For these hard queries, the graph sampling provides reasonably accurate estimates. Even for the aforementioned extreme case, the graph sampling achieves a q-error of about 1.6.

\section{Conclusion}
\label{sec:conclusion}
In this paper, we have proposed \Fastest for subgraph cardinality estimation. Our novel filtering-sampling approach synergistically combines (1) a strong filtering method for drastically reducing the sample space, (2) an efficient and accurate estimation by candidate tree sampling, and (3) a worst-case optimal stratified graph sampling with outstanding accuracy on hard instances. 
Extensive experiments on real-world datasets show that \Fastest significantly outperforms state-of-the-art algorithms in terms of accuracy and time, while maintaining comparable memory usage.

\section*{Acknowledgements}
\noindent W. Shin, S. Song, and K. Park were supported by Institute of Information communications Technology Planning Evaluation (IITP) grant funded by the Korea government (MSIT) (No. 2018-0-00551, Framework of Practical Algorithms for NP-hard Graph Problems). W.-S. Han was supported by Institute of Information communications Technology Planning Evaluation (IITP) grant funded by the Korea government (MSIT) (No. 2021-0-00859, Development of a Distributed Graph DBMS for Intelligent Processing of Big Graph).

\newpage
\bibliographystyle{ACM-Reference-Format}
\bibliography{bib}
\newpage
\nobalance
\appendix
\section{Proofs and Discussions}

\label{sec:appendix}
\subsection{Proof of \Cref{thm:total_complexity}}
\begin{lemma}
\label{lem:num_cycles}
    The number of triangles and 4-cycles in a graph $g$ is at most $O(\abs{E_g}\delta_g)$ and $O(\abs{E_g} \Delta_g \delta_g)$, respectively.
\end{lemma}
\begin{proof}
    For an edge $(v, v') \in E_g$, the number of triangles containing $(v, v')$ is at most $\min(d_g(v), d_g(v'))$. Chiba and Nishizeki \cite{ChibaNishizeki} have shown that the total number of triangles is at most
    \begin{equation}
        \sum_{(v, v') \in E_g} \min(d_g(v), d_g(v')) = O(\abs{E_g} \delta_g).
    \end{equation}

    Similarly, the number of 4-cycles containing $(v, v')$ is at most $d_g(v)d_g(v')$. Since $d_g(v)d_g(v') \leq \Delta_g\min(d_g(v), d_g(v'))$, 
    \begin{align*}
        \sum_{(v, v') \in E_g} d_g(v) d_g(v') 
        &\leq \Delta_g \sum_{(v, v') \in E_g} \min(d_g(v), d_g(v')) \\
        &= O(\abs{E_g} \Delta_g \delta_g)
    \end{align*}
    Therefore the number of 4-cycles is at most $O(\abs{E_g} \Delta_g \delta_g)$.
\end{proof}
\begin{lemma}
\label{thm:time_complexity_triangle}
    The time complexity for Triangle Safety applied on each query edge $(u, u')$ is $O(\min(d_q(u), d_q(u')) \abs{E_G} \delta_G)$.
\end{lemma}

\begin{proof}
    For a fixed $(u, u') \in E_q$, the number of triangles containing $(u, u')$ is at most $\min(d_q(u), d_q(u'))$. For each $u^* \in L^3_q(u, u')$, checking condition 2 in \Cref{def:trianglesafety} takes time linear to the number of all triangles in the data graph in the worst case. Therefore, the said complexity is achieved by \Cref{lem:num_cycles}.
\end{proof}
 
\begin{lemma}
\label{thm:time_complexity_four_cycle}
    The time complexity for Four-Cycle Safety applied on each query edge $(u, u')$ is $O(d_q(u) d_q(u') \abs{E_G} \Delta_G \delta_G)$.
\end{lemma}
\begin{proof}
    The proof can be done analogously, using the fact that the number of 4-cycles containing $(u, u')$ is at most $d_q(u) d_q(u')$ and the total number of 4-cycles in the data graph is bounded by $O(\abs{E_G}\Delta_G\delta_G)$ as in \autoref{lem:num_cycles}.
\end{proof}

\ThmTotalComplexity*

\begin{proof}
    Let the sequence of query vertices chosen to be refined be $s = (u^k), u^k \in V_q$. The second stopping criterion ensures $\sum_{u \in s} d_q(u)\allowbreak \leq R\abs{E_q}$. For each $u \in V_q$, candidate set $C(u)$ is a subset of $V_G$. 
    Hence for Edge Bipartite Safety, we have:
    \begin{equation*}
        \sum_{u \in s} \sum_{v \in C(u)} d_q(u)^2 d_G(v) \leq \Delta_q \sum_{u \in s} d_q(u) \abs{E_G} \leq R\abs{E_q}\abs{E_G}\Delta_q.
    \end{equation*}
    Since $R$ is a constant, this is simplified to $O(\abs{E_q}\abs{E_G}\Delta_q)$.

    For Triangle Safety, we have:
    \begin{equation*}
        \sum_{u \in s} \sum_{u' \in N_q(u)} \min(d_q(u), d_q(u')) \abs{E_G} \delta_G = O(\abs{E_q}\abs{E_G} \Delta_q \delta_G)
    \end{equation*}
    as the total number of edges seen during the refinement is bounded by $O(\abs{E_q})$. Analysis of 4-Cycle Safety can be done analogously.
\end{proof}
\vfill{\eject}

\subsection{Proof of Theorem~\ref{thm:density_heuristic_optimality}}
\label{prf:expected_candidate_tree}
In this section, we assume that, for all $(v, v') \in C(u) \times C(u')$, the event that edge $(v, v')$ is a candidate for $(u, u')$ has the probability ${density}(u, u')$, and it is independent from the events that other vertex pairs are candidate edges.
For computing the expected number of candidate trees, we first present a lemma for the expected value of $D(u,v)$ in \autoref{eq:numtree-recurrence}. 

Given query graph $q$ and $u\in V_q$, let $T_u$ be a subtree of $T_q$ rooted at $u$, and $ch(u)$ be the set of child vertices of $u$.
\begin{lemma}\label{lem:expected_D}
Under the above assumption, the expected value of $D(u, v)$ for the spanning tree $T_q$ is:
\begin{equation} \label{eq:expected_D}
    \Expectation{D(u, v)} = \frac{\prod_{(u_i, u_j) \in E_{T_u}} \abs{E_{CS}(u_i, u_j)}}{\prod_{u'\in V_{T_u}} \abs{C(u')}^{d_{T_u}(u')-1}}\times\frac{1}{\abs{C(u)}}.
\end{equation}
\end{lemma}

\begin{proof}
By induction.
\autoref{eq:expected_D} clearly holds when $\abs{V_{T_u}}=1$ as $\Expectation{D(u,v)}=1$. Otherwise, $u$ have child vertices $x_1, x_2,\dots,x_k$, and the expected value of $D(u,v)$ is
\begin{align}
    &\mathbb{E}\left[\prod_{x\in ch(u)} \sum_{v_c \in C(x\mid u, v)} D(x, v_c)\right]\notag\\
    &\overset{\scalebox{0.7}{(1)}}{=}\prod_{x\in ch(u)} \mathbb{E}\left[\sum_{v_c \in C(x\mid u, v)} D(x, v_c)\right]\notag\\
    &=\prod_{x\in ch(u)} \sum_{v_c \in C(x)} \Prob{v_c \in C(x \mid u, v)}\times \mathbb{E}\left[D(x, v_c)\right]\notag\\    
    &=\prod_{x\in ch(u)} \sum_{v_c \in C(x)} \Prob{(v, v_c) \in E_{CS}(u, x)}\times \mathbb{E}\left[D(x, v_c)\right]\notag\\
    &=\prod_{x\in ch(u)} \sum_{v_c \in C(x)} density(u,x)\times \mathbb{E}\left[D(x, v_c)\right]\notag\\
    &\overset{\scalebox{0.7}{(2)}}{=}\prod_{x\in ch(u)} density(u,x)\sum_{v_c \in C(x)} \frac{\prod_{(u_i, u_j) \in E_{T_x}} \abs{E_{CS}(u_i, u_j)}}{\prod_{u'\in V_{T_x}} \abs{C(u')}^{d_{T_x}(u')-1}}\times\frac{1}{\abs{C(x)}}\notag\\
    &=\prod_{x\in ch(u)} \frac{\abs{E_{CS}(u, x)}}{\abs{C(u)}\abs{C(x)}}\times \frac{\prod_{(u_i, u_j) \in E_{T_x}} \abs{E_{CS}(u_i, u_j)}}{\prod_{u'\in V_{T_x}} \abs{C(u')}^{d_{T_x}(u')-1}}\notag\\
    &=\frac{\prod_{(u_i, u_j) \in E_{T_u}} \abs{E_{CS}(u_i, u_j)}}{\prod_{u'\in V_{T_u}} \abs{C(u')}^{d_{T_u}(u')-1}}\times\frac{1}{\abs{C(u)}}.
\end{align}
where (1) holds from the independence assumption, and (2) holds from the inductive hypothesis (\autoref{eq:expected_D}).
\end{proof}

\vfill{\eject}
We can get the expected total number of candidate trees by \Cref{lem:expected_D}.

\begin{lemma}\label{lem:expected_candidate_tree}
Under the above assumption, the expected value of the number of candidate trees is:
\begin{equation} \label{eq:expected_candidate_trees}
    \frac{\prod_{(u_i, u_j) \in E_{T_q}} \abs{E_{CS}(u_i, u_j)}}{\prod_{u\in V_q} \abs{C(u)}^{d_{T_q}(u)-1}}.
\end{equation}
\end{lemma}

\begin{proof}
The expected total number of candidate trees is the expectation of the sum of the number of candidate trees that the root $u_r$ is mapped to each candidate vertex $v\in C(u_r)$. 
\begin{align*}\label{eq:expected_tree}
    \mathbb{E}\left[\sum_{v\in C(u_r)} D(u_r, v)\right]
    &=\sum_{v\in C(u_r)} \frac{\prod_{(u_i, u_j) \in E_{T_q}} \abs{E_{CS}(u_i, u_j)}}{\prod_{u\in V_{T_q}} \abs{C(u)}^{d_{T_q}(u)-1}}\times\frac{1}{\abs{C(u_r)}}\notag\\
    &=\frac{\prod_{(u_i, u_j) \in E_{T_q}} \abs{E_{CS}(u_i, u_j)}}{\prod_{u\in V_q} \abs{C(u)}^{d_{T_q}(u)-1}} 
\end{align*}
\end{proof}

\ThmDensityHeuristicOptimality*

\begin{proof}
When the weight of an edge of a tree is defined as density, the product of edge weights is
\begin{equation}
\label{eq:prod_density}
    \prod_{(u_i, u_j) \in E_{T_q}}  \frac{\abs{E_{CS}(u_i, u_j)}}{\abs{C(u_i)}\abs{C(u_j)}}
    =\frac{\prod_{(u_i, u_j) \in E_{T_q}} \abs{E_{CS}(u_i, u_j)}}{\prod_{u\in V_q} \abs{C(u)}^{d_{T_q}(u)}}.
\end{equation} 
Since $\displaystyle\prod_{u\in V_q} \abs{C(u)}$ is constant for the fixed candidate space, 
the spanning tree minimizing \autoref{eq:prod_density} also minimizes \autoref{eq:expected_candidate_trees}, which is the expected number of candidate trees.
\end{proof}

\subsection{Counting Candidate Trees}
Let $T_u$ be a subtree of $T_q$ rooted at $u \in V_q$, and $D(u, v)$ be the number of candidate trees for $T_u$ with $u$ mapped to $v$.
The number of candidate trees follows the recurrence
\begin{equation}
\label{eq:numtree-recurrence-appendix}
    D(u, v) = \prod_{u_c \in \text{children of } u} \ \sum_{v_c \in C(u_c\mid u,v)} D(u_c, v_c)
\end{equation}
as stated in Section 5.2.
The recurrence can be computed by dynamic programming as in \Cref{alg:TreeWeight}.
\begin{algorithm}[t]
    \rm\sffamily
    \caption{\textsc{CountCandidateTrees}(\textsf{CS}, $q, G, T_q$)}
    \label{alg:TreeWeight}
    \KwIn{The compact \textsf{CS}, a spanning tree $T_q$ of $q$}
    
    \ForEach{\rm\sffamily $u$ in postorder traversal of $T_q$} {
        \ForEach{\rm\sffamily $v \in C(u)$} {
            $D(u, v)$ $\gets$ 1\;

            \ForEach{\rm\sffamily Child node $u_c$ of $u$ in $T_q$} {
                child\_choices $\gets$ $0$\; 
                \ForEach{\rm $v_c \in C(u_c \mid u, v)$} {
                    child\_choices $\gets$ child\_choices $+ D(u_c, v_c)$ \; 
                }
                $D(u, v)$ $\gets$ $D(u, v) \times$ child\_choices
            }
        }
    }
    \Return{D}
\end{algorithm}

\begin{restatable}{theorem}{ThmTimeComplexityTreeDP}
\label{thm:time_complexity_tree_DP}
    \Cref{alg:TreeWeight} runs in $O(\abs{E_q}\abs{E_G})$ time. 
\end{restatable}

\begin{proof}
    Line 7 in \Cref{alg:TreeWeight} is executed at most once for every edge in the Candidate Space, and the number of candidate edges is at most $O(\abs{E_q}\abs{E_G})$.
\end{proof}

\subsection{Proof of \Cref{thm:tree_uniform}}

\ThmTreeUniform*

\begin{proof}
    Let $S$ be a random variable denoting a candidate tree and $s$ be a sample candidate tree.
    We denote $s(u_i)$ as $s_i$, and the parent of $u_i$ as $u^{p}_{i}$ for $u_i \in V_q$ except the root $u_r$.    
    The probability of the candidate tree $s$, $\Prob{S = s}$, is written as 
    \begin{equation}
        \Prob{S = s} = \Prob{S(u_r)=s_r}\prod_{i\neq r}\Prob{S(u_i)=s_i\mid S(u^{p}_{i})=s^{p}_{i}}
    \end{equation}
    as the probability of $S(u_i)=s_i$ depends only on its parent $u^{p}_{i}$. 

    \Cref{alg:RandomTree} assigns each probability as: 
    \begin{align*}
        \Prob{S(u_r)=s_r} &= \frac{D(u_r, s_r)}{\sum_{v \in C(u_r)} D(u_r, v)}\\ 
        \Prob{S(u_i)=s_i\mid S(u^{p}_{i})=s^{p}_{i}} &= \frac{D(u_i, s_i)}{\sum_{v \in C(u_i\mid u^{p}_{i}, s^{p}_{i})} D(u_i, v)}
    \end{align*}
    Hence the probability is 
    \begin{equation*}
        \Prob{S = s} = \frac{D(u_r, s_r)}{\sum_{v \in C(u_r)} D(u_r, v)} \prod_{u_i \neq u_r} \frac{D(u_i, s_i)}{\sum_{v \in C(u_i\mid u^{p}_{i}, s^{p}_{i})} D(u_i, v)}
    \end{equation*}
    By rearranging the terms, the above probability becomes
    \begin{align*}
        &\frac{\prod_{u_i \in V_q} D(u_i, s_i) }{\sum_{v \in C(u_r)} D(u_r, v)}\times{\underbrace{\left(\prod_{u_i \neq u_r} \sum_{v \in C(u_i\mid u^{p}_{i}, s^{p}_{i})} D(u_i, v)\right)}_{(1)}}^{-1}.
    \end{align*}
    (1) can be reformulated by double counting as 
    \begin{align*}
        (1) &= \prod_{u_i \in V_q}\prod_{u_j \in ch(u_i)} \sum_{v \in C(u_j\mid u_i, s_i)} D(u_j, v).
    \end{align*}
    It is simplified by \autoref{eq:numtree-recurrence-appendix} as
    \begin{align*}
        (1) = \prod_{u_i \in V_q} D(u_i, s_i).
    \end{align*}
    Hence the given probability is 
    \begin{align}
        \Prob{S = s} = \frac{1}{\sum_{v \in C(u_r)} D(u_r, v)}
    \end{align}
    which is invariant to the choice of $s$.
\end{proof}
\vfill{\eject}
\subsection{Proof of \Cref{thm:tree_unbiased}}
\ThmTreeUnbiased*

\begin{proof}
    Unbiasedness is immediate from the general sampling framework, since candidate trees are a superset of the counterparts of embeddings and \Cref{alg:RandomTree} returns a uniform random sample (\autoref{thm:tree_uniform}). 
    
    Consistency follows from the Chernoff bound. Let $\hat{\rho}$ be the estimate of $\rho$ given by \Cref{alg:treesampling}, taking $\#t$ samples. By the Chernoff bound \cite{MeasureConcentration}, we have
    \begin{align*}
        \Prob{\frac{\abs{\hat{\rho} - \rho}}{\rho} > (c - 1)} < 2 e^{-(c - 1)^2 (\#t) \rho / 3}
    \end{align*}
    for any $c > 1$. As $\#t \to \infty$, the probability of the relative error being larger than $c-1$ converges to 0, making the estimator consistent.
\end{proof}

\subsection{Discussion on the Clopper-Pearson Interval}
There are other intervals we can use instead of the Clopper-Pearson interval such as the standard interval, Wilson score interval, and Chernoff bound \cite{IntervalEstimation, WilsonScoreInterval, Chernoff, GJSampler}.
We select the Clopper-Pearson because of the following reasons. The standard interval, which relies on approximation to the standard normal distribution, is known to perform poorly in cases where the probability of success is close to 0.
Such cases do occasionally occur when the number of candidate trees is too large when compared to the number of embeddings.
In addition, it is guaranteed that the probability that the Clopper-Pearson interval for the $1-\alpha$ confidence level covers $\rho$ is no less than $1-\alpha$ unlike the standard interval and Wilson score interval.
The Chernoff bound~\cite{Chernoff} is commonly used by other sampling algorithms, and offers the advantage of requiring a constant number of successes to satisfy $\Prob{c^{-1} \rho\le \hat{\rho}\le c \rho}\ge 1-\alpha$~\cite{GJSampler}.
However, the Clopper-Pearson interval tends to require fewer trials than the Chernoff bound.
Empirically, for $\alpha=0.05$ and $c=1.25$, we discovered that the stated condition is met with fewer than 88 successes for the Clopper-Pearson interval, compared to more than 140 successes for the Chernoff bound.

\subsection{Proof of \Cref{thm:graph_unbiased}}

\ThmGraphUnbiased*

\begin{proof}
We prove that $\hat{w}_M$ defined in \autoref{eq:RecursiveEstimate} is an unbiased estimator for $w_M$, regardless of the value $ub$ for each group by induction. It is clear that $\Expectation{\hat{w}_M} = w_M$ when $\abs{M} = \abs{V_q}$ or $C_M(u) = \emptyset$. Assume that for any $ub > 0$, $\hat{w}_{M \cup \set{(u, v)}}$ is an unbiased estimate for $w_{M \cup \set{(u, v)}}$. Note that sizes of $C_M(u)$ and $S$ are both fixed. By \autoref{eq:RecursiveEstimate}, we have 
\begin{align*}
    \Expectation{\hat{w}_M} &= \Expectation{\frac{\abs{C_M(u)}}{\abs{S}} \sum_{v \in S} \hat{w}_{M \cup (u, v)}}\\ 
    &= \frac{\abs{C_M(u)}}{\abs{S}} \Expectation{\sum_{v \in S} \hat{w}_{M \cup (u, v)}}\\
    &\overset{\scalebox{0.7}{(1)}}{=} \frac{\abs{C_M(u)}}{\abs{S}} \Expectation{\Expectation{\sum_{v \in S} \hat{w}_{M \cup (u, v)} \mid S}}\\ 
    &\overset{\scalebox{0.7}{(2)}}{=} \frac{\abs{C_M(u)}}{\abs{S}} \Expectation{\sum_{v \in S} w_{M \cup (u, v)}}\\ 
    &= \frac{\abs{C_M(u)}}{\abs{S}} \sum_{S} \Prob{S}\sum_{v \in S} w_{M \cup (u, v)}\\ 
    &= \frac{\abs{C_M(u)}}{\abs{S}} \sum_{S} \sum_{v \in S} \Prob{S} w_{M \cup (u, v)}\\ 
    &\overset{\scalebox{0.7}{(3)}}{=}  \frac{\abs{C_M(u)}}{\abs{S}} \sum_{v \in C_M(u)} \sum_{S \text{ containing } v} \Prob{S}  w_{M \cup (u, v)}\\ 
    &\overset{\scalebox{0.7}{(4)}}{=} \frac{\abs{C_M(u)}}{\abs{S}} \sum_{v \in C_M(u)} \sum_{S \text{ containing } v} \binom{\abs{C_M(u)}}{\abs{S}}^{-1}  w_{M \cup (u, v)}\\
    &\overset{\scalebox{0.7}{(5)}}{=} \frac{\abs{C_M(u)}}{\abs{S}} \sum_{v \in C_M(u)} \binom{\abs{C_M(u)}-1}{\abs{S}-1}\binom{\abs{C_M(u)}}{\abs{S}}^{-1}  w_{M \cup (u, v)}\\
    &= \sum_{v \in C_M(u)} w_{M \cup (u, v)}
\end{align*}
where (1) holds by the law of total expectation, (2) holds by the inductive hypothesis and linearity of expectation, (3) is a double counting argument, and (4) and (5) holds as $S$ is a random subset of size $\abs{S}$ of $C_M(u)$.
Therefore, by induction, $\hat{w}_{\emptyset}$ is an unbiased estimator for $w_{\emptyset}$, which is the number of embeddings.

For sufficiently large $ub_{\emptyset}$, the algorithm can use $S = C_M(u)$ to examine every $v \in C_M(u)$. This is then equivalent to a backtracking algorithm for subgraph matching, and therefore \Cref{alg:CGSampling} is consistent. Backtracking with the set intersection for computing candidates can be seen as the worst-case optimal join algorithm \cite{WCOJoin, WCOSubgraphQueryProcessing}, hence the complexity is bounded by $O(AGM(q))$. 
\end{proof}
\vfill{\eject}
\end{document}